\documentclass[11pt,letterpaper]{article}
\usepackage[margin=1in]{geometry}

\usepackage[utf8]{inputenc}

\usepackage{amsmath}
\usepackage{amsthm}
\usepackage{amssymb}
\usepackage{fullpage}
\usepackage{paralist}
\usepackage[normalem]{ulem}

\usepackage[compact]{titlesec}
\titlespacing*{\section}{0pt}{2ex}{1ex}
\titlespacing*{\subsection}{0pt}{1.5ex}{0.5ex}
\renewcommand{\paragraph}[1]{\vspace{2pt}\noindent\textbf{#1}. }
\makeatletter%
\renewenvironment{proof}[1][\proofname]{\par%
\pushQED{\qed}%
\normalfont \topsep1\p@\@plus1\p@\relax%
\trivlist%
\item\relax%
{\itshape%
#1\@addpunct{.}}\hspace\labelsep\ignorespaces%
}{%
\popQED\endtrivlist\@endpefalse%
}%
\makeatother

\usepackage{subcaption}
\usepackage{enumitem}
\usepackage{colortbl}
\usepackage[bookmarks]{hyperref}

\usepackage{tikz}
\usepackage{nicematrix}
\usepackage{nicefrac}
\usepackage[nodisplayskipstretch]{setspace}
\AtBeginDocument{
\setlength{\abovedisplayskip}{2.5pt}
\setlength{\abovedisplayshortskip}{2.5pt}
\setlength{\belowdisplayskip}{2.5pt}
\setlength{\belowdisplayshortskip}{2.5pt}
}

\swapnumbers
\numberwithin{equation}{section}
\newtheorem{definition}[equation]{Definition}
\newtheorem{example}[equation]{Example}
\newtheorem{lemma}[equation]{Lemma}
\newtheorem{proposition}[equation]{Proposition}
\newtheorem{theorem}[equation]{Theorem}
\newtheorem{remark}[equation]{Remark}
\newtheorem{corollary}[equation]{Corollary}
\newtheorem{claim}[equation]{Claim}

\newcommand{\un}{\mathsf{u}}
\newcommand{\Bool}{\{0, 1\}}
\newcommand{\Tri}{\{0, \un, 1\}}
\newcommand{\MUX}{\mathsf{MUX}}
\newcommand{\der}{\mathsf{d}}

\newcommand{\kw}[1]{\mathsf{KW}^{#1}}
\newcommand{\entry}[1]{#1}
\newcommand{\bin}{\mathrm{bin}}
\newcommand{\size}[1]{\mathsf{size}^{#1}}
\newcommand{\depth}[1]{\mathsf{depth}^{#1}}
\newcommand{\cc}{\mathsf{CC}}
\newcommand{\hfe}[1]{\widetilde#1}
\newcommand{\psize}{\mathsf{monorect}}
\newcommand{\monorect}{\mathsf{monorect}}
\newcommand{\rank}{\mathsf{rank}}

\title{Karchmer-Wigderson Games for Hazard-free Computation}
\author{Christian Ikenmeyer\thanks{University of Warwick, United Kingdom. email: christian.ikenmeyer@warwick.ac.uk}, Balagopal Komarath\thanks{IIT Gandhinagar, India. email: bkomarath@rbgo.in}, Nitin Saurabh\thanks{IIT Hyderabad, India. email: nitin@cse.iith.ac.in. This project has received funding from the the European Union’s Horizon 2020 research and innovation programme under grant agreement No 802020-ERC-HARMONIC.}}
\date{November 2022}

\begin{document}
\raggedbottom

\maketitle

\thispagestyle{empty}
\setcounter{page}{0}

\begin{abstract}
We present a Karchmer-Wigderson game to study the complexity of hazard-free formulas. This new game is both a generalization of the monotone Karchmer-Wigderson game and an analog of the classical Boolean Karchmer-Wigderson game. Therefore, it acts as a bridge between the existing monotone and general games.

Using this game, we prove hazard-free formula size and depth lower bounds that are provably stronger than those possible by the standard technique of transferring results from monotone complexity in a black-box fashion.
For the multiplexer function
we give (1) a hazard-free formula of optimal size and (2) an improved low-depth hazard-free formula of almost optimal size and (3) a hazard-free formula with alternation depth 2 that has optimal depth.
We then use our optimal constructions to obtain an improved universal worst-case hazard-free formula size upper bound.
We see our results as a step towards establishing hazard-free computation as an independent missing link between Boolean complexity and monotone complexity.
\end{abstract}

\bigskip\bigskip

\noindent\begin{minipage}{2cm}
\end{minipage}
\begin{minipage}{15cm}
\footnotesize
Hazard-free computation, monotone computation, Karchmer-Wigderson games, communication complexity, lower bounds
\end{minipage}

\medskip

\noindent\begin{minipage}{4.5cm}
\footnotesize
\textbf{ACM subject classification:}\\\mbox{~}\\\mbox{~}\\\mbox{~}
\end{minipage}
\begin{minipage}{12.5cm}
\footnotesize
Theory of Comput.\ $\rightarrow$ Models of Computation $\rightarrow$ Concurrency \\
Theory of Comput.\ $\rightarrow$ Comput.\ Compl.\ and Cryptogr.\ $\rightarrow$ Communication Complexity\\
Theory of Comput.\ $\rightarrow$ Comput.\ Compl.\ and Cryptogr.\ $\rightarrow$ Circuit Complexity\\
Hardware $\rightarrow$ Integrated Circuits $\rightarrow$ Logic Circuits $\rightarrow$ Combinational Circuits
\end{minipage}

\bigskip\bigskip

\noindent\begin{minipage}{3.5cm}
\footnotesize
\textbf{Acknowledgments:}\\\mbox{~}\\\mbox{~}
\end{minipage}
\begin{minipage}{13cm}
\footnotesize
We thank Igor Sergeev for his support with the literature. 
Nitin Saurabh would like to thank Yuval Filmus for helpful discussions. We would also like to thank anonymous reviewers for useful suggestions that greatly improved the presentation. 
\end{minipage}

\newpage

\section{Introduction}
In this paper we apply for the first time methods from communication complexity to the study of hazard-free complexity, which we see as a step towards bridging the gap between Boolean complexity and monotone complexity.

The study of the three-valued strong logic of indeterminacy dates back to Kleene (\cite[p.~153]{Kle:38}, \cite[\S64]{Kle:52}). It found numerous applications, for example in logic (see e.g.\ \cite{KOR:66}, \cite{MAL:14}), in cybersecurity for information flow tracking at the gate level (see e.g.\ \cite{TWMMCS:09}, \cite{HOITSMK:12}, \cite{BHTBKI:17}), the design of real-world circuits that communicate between unsynchronized clock domains (see e.g.\ \cite{FFL:18}, \cite{FKLP:17}, \cite{TFL:17}, \cite{BLM:20}), and in the study of hazards in Boolean circuits (see e.g.\ \cite{Got:49,CAL:58,YR:64,EIC:65,MUK:72,MUK:83B,MUK:83A,ND:92,BS:95,BEI:01}).
The languages in these areas is different, but the underlying three-valued logic is the same and many questions and results can be readily transferred between areas.
We will use the language of hazards in circuits in this paper.
The use of three-valued logic to study hazards in Boolean circuits dates all the way back to Goto \cite{Got:49}, who used 0 and 1 to denote the Boolean values and used
the symbol $\frac 1 2$ to denote the third value, which stands for any undefined, oscillating, unstable, or otherwise somehow flawed state.
In this paper we use the symbol $\un := \frac 1 2$ to denote this third state.
Goto modeled the Boolean operations $\wedge$ (and) and $\vee$ (or) as \textsf{min} and \textsf{max}, respectively, and the $\neg$ (not) operation as $1-x$, which defines the behaviour of the three types of gates on inputs from $\{0,\un,1\}$.
Hence a Boolean circuit $C$ on $n$ inputs\footnote{All circuits in our paper have a single output.} computes a function $\{0,\un,1\}^n \to \{0,\un,1\}$ by induction over the circuit structure.
The design of the gate behaviour as \textsf{min}, \textsf{max}, and $1-x$ is the result of a more general construction principle that is called the \emph{hazard-free extension}\footnote{The function $\hfe{f}$ is called the \emph{hazard-free extension} of $f$ (see \cite{IKLLMS19}), or alternatively the \emph{ternary extension} (see \cite{MSB:12}) or the \emph{metastable closure} (see \cite{FFL:18}).}
$\hfe{f}:\{0,\un,1\}^n\to\{0,\un,1\}$ of a Boolean function $f:\{0,1\}^n\to\{0,1\}$.
It is defined as follows.

A binary string $a \in \{0,1\}^n$ is called a \emph{resolution} of a ternary string $\alpha \in \{0,\un,1\}^n$ if for all $1\leq i\leq n$ with $\alpha_i\neq\un$ we have $\alpha_i = a_i$, i.e., all entries $\un$ are replaced by 0s and 1s.
Note that the set of all resolutions $a$ of $\alpha$ forms a subcube of $\{0,1\}^n$.
For a Boolean function $f : \{0,1\}^n \to \{0,1\}$ and for an input $\alpha \in \{0,\un,1\}^n$ we define
the evaluation of the function $\hfe{f} : \{0,\un,1\}^n\to\{0,\un,1\}$ at~$\alpha$ via
\begin{equation}\label{eq:defhazfreeextension}
\hfe{f}(\alpha) := \begin{cases}
1 & \text{ if for all resolutions $a$ of $\alpha$ we have } f(a) = 1
\\
0 & \text{ if for all resolutions $a$ of $\alpha$ we have } f(a) = 0
\\
\un & \text{ otherwise}.
\end{cases}
\end{equation}
A Boolean circuit $C$ that computes a Boolean function $f:\{0,1\}^n\to\{0,1\}$ is called \emph{hazard-free} if
for all $\alpha\in\{0,\un,1\}^n$ we have $C(\alpha)=\hfe{f}(\alpha)$.
An $\alpha$ where these two functions differ is called a \emph{hazard}.
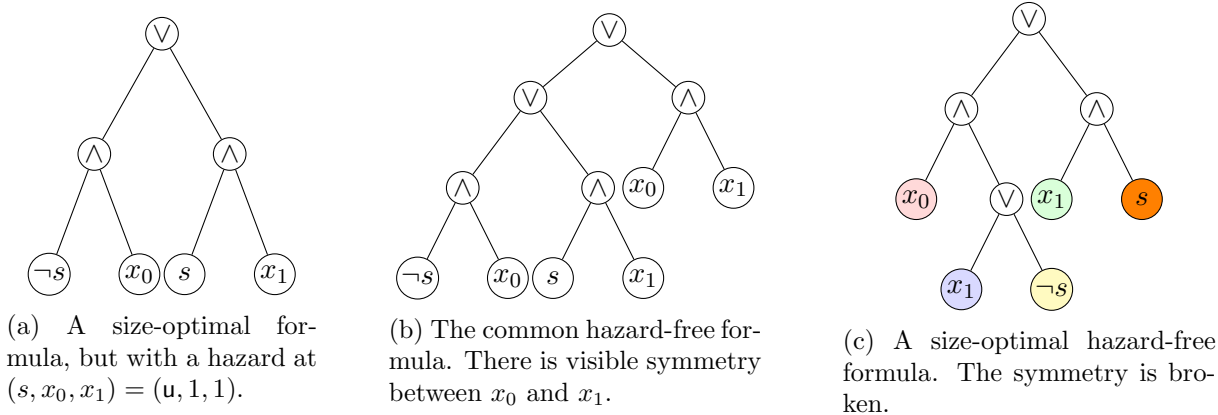
\begin{figure}
    \centering
\begin{subfigure}{0.25\textwidth}
\centering
    \begin{tikzpicture}[vertex/.style={circle,draw,minimum size=1em,inner sep=1pt}, xscale=0.6, yscale=0.8]
    
    \node[vertex] (1) at (0,0) {$\vee$};
    \node[vertex] (2) at (-1.5,-2) {$\wedge$};
    \node[vertex] (3) at (1.5,-2) {$\wedge$};
    \node[vertex] (4) at (-2.5,-4) {$\neg s$};
    \node[vertex] (5) at (-0.5,-4) {$x_0$};
    \node[vertex,inner sep=3pt] (6) at (0.5,-4) {$s$};
    \node[vertex] (7) at (2.5,-4) {$x_1$};
    
    \draw[] (1) -- (2);
    \draw[] (1) -- (3);   
    \draw[] (2) -- (4);
    \draw[] (2) -- (5); 
    \draw[] (3) -- (6);
    \draw[] (3) -- (7);
    \end{tikzpicture}
    \caption{A size-optimal formula, but with a hazard at $(s,x_0,x_1)=(\un,1,1)$.}
\end{subfigure}%
\hspace{0.05\textwidth}
\begin{subfigure}{0.3\textwidth}
\centering
    \begin{tikzpicture}[vertex/.style={circle,draw,minimum size=1em,inner sep=1pt}, xscale=0.6, yscale=0.6]
    
    \node[vertex] (1) at (0,0) {$\vee$};
    \node[vertex] (2) at (-1.5,-2) {$\wedge$};
    \node[vertex] (3) at (1.5,-2) {$\wedge$};
    \node[vertex] (4) at (-2.5,-4) {$\neg s$};
    \node[vertex] (5) at (-0.5,-4) {$x_0$};
    \node[vertex,inner sep=3pt] (6) at (0.5,-4) {$s$};
    \node[vertex] (7) at (2.5,-4) {$x_1$};
    \node[vertex] (8) at (1.75,1.5) {$\vee$};
    \node[vertex] (9) at (3.5,0) {$\wedge$};
    \node[vertex] (10) at (2.5,-2) {$x_0$};
    \node[vertex] (11) at (4.5,-2) {$x_1$};
    
    \draw[] (1) -- (2);
    \draw[] (1) -- (3);   
    \draw[] (2) -- (4);
    \draw[] (2) -- (5); 
    \draw[] (3) -- (6);
    \draw[] (3) -- (7);
    \draw[] (1) -- (8);
    \draw[] (8) -- (9);
    \draw[] (9) -- (10);
    \draw[] (9) -- (11);
    \end{tikzpicture}
    \caption{The common hazard-free formula. There is visible symmetry between $x_0$ and $x_1$.}
\end{subfigure}
\hspace{0.05\textwidth}
\begin{subfigure}{0.3\textwidth}
\centering
    \begin{tikzpicture}[xscale=0.6, yscale=0.8]
    \node [circle,draw,minimum size=1em,inner sep=1pt] {$\vee$} [sibling distance=3cm]
      child { node [circle,draw,minimum size=1em,inner sep=1pt] {$\wedge$} [sibling distance=2cm]
        child {node [circle,draw,fill=red!15,minimum size=1em,inner sep=1pt] {$x_0$}}
        child {node [circle,draw,minimum size=1em,inner sep=1pt] {$\vee$}
          child {node [circle,draw,fill=blue!15,minimum size=1em,inner sep=1pt] {$x_1$}}
          child {node [circle,draw,fill=yellow!30,minimum size=1em,inner sep=1pt] {$\neg s$}}}}
      child {node [circle,draw,minimum size=1em,inner sep=1pt] {$\wedge$} [sibling distance=2cm]
        child {node [circle,draw,fill=green!15,minimum size=1em,inner sep=1pt] {$x_1$}}
        child {node [circle,draw,fill=orange,minimum size=1em,inner sep=3pt] {$s$}}};
    \end{tikzpicture}
    \caption{A size-optimal hazard-free formula. The symmetry is broken.}
    \label{fig:mux_formula}
\end{subfigure}
    \caption{Different De Morgan formulas for $\MUX_1$}
    \label{fig:mux1-eg}
\end{figure}
For example, consider the circuit in Part~(a) of Figure~\ref{fig:mux1-eg} that computes the multiplexer function $C(s,x_0,x_1)=\MUX(s,x_0,x_1)=x_s$ for all $(s,x_0,x_1)\in\{0,1\}^3$. We observe that $C$ has a hazard at $(\un,1,1)$, because $C(\un,1,1)=\un \vee \un = \un$, whereas $C(0,1,1)=C(1,1,1)=1$.
The circuit can be made hazard-free at the expense of using more gates, see Part~(b) of Figure~\ref{fig:mux1-eg} (this construction can be found for example in \cite[Fig.~6a]{FFL:18} and \cite[Fig.~1b]{IKLLMS19}).

Designing small hazard-free circuits for computing Boolean functions is a fundamental goal in electronic circuit design.
Huffman \cite{Huf:57} proved that all Boolean functions can be implemented by hazard-free circuits and he already noted the large growth of the number of gates in his examples.
Eichelberger proved the first lower bound on hazard-free complexity in the restricted model of DNF formulas, which is given by the number of prime implicants of the function that is computed.
The very recent paper \cite{IKLLMS19} formally defines the notion of hazard-free complexity and shows that for monotone functions the hazard-free complexity and the monotone complexity coincide.
Fortunately, good lower bounds are known on the monotone complexity of monotone Boolean functions (see \cite{R:85a, R:85b, A:87, AB:87, R:87, T:88, GS:85, KW90, RW:92, RM:99, HR:00, GP:14,PR:17}).
A direct consequence of \cite{IKLLMS19} is that the exponential gap between Boolean circuit complexity and monotone circuit complexity
transfers directly into an exponential gap between Boolean circuit complexity and the hazard-free circuit complexity.
\cite{Jukna21} proves that every Boolean circuit that computes a monotone function and that is optimal with respect to hazard-free complexity must automatically be a monotone circuit.
Hence the study of hazard-free complexity does not yield any new insights into monotone functions, but it is a natural generalization of monotone complexity to the domain of \emph{all} Boolean functions.
This suggests that the study of hazard-free complexity, in particular of non-monotone functions, should be of independent interest (apart from its applicability in practice).
As a first step in this direction, \cite{IKLLMS19} prove lower bounds for non-monotone functions by using monotone circuit lower bounds for the \emph{hazard-derivative} of the function, because the monotone complexity of the hazard-derivative of $f$ is a lower bound on the hazard-free complexity of~$f$.
All existing lower bounds known for hazard-free computation \cite{IKLLMS19, Jukna21} 
are derived from this wealth of known monotone complexity lower bounds.

However, the hazard-derivative method cannot always prove optimal lower bounds, because some functions with high hazard-free complexity have hazard-derivatives of only low monotone complexity (compare Proposition~\ref{prop:derivatives} with Theorem~\ref{thm:mux-lower}).
We call this problem the \emph{monotone barrier}.
In this paper we take a radically different approach than all previous papers and translate notions from communication complexity to the hazard-free setting.
The result is a new type of the Karchmer-Wigderson game that exactly describes the hazard-free De Morgan formula size and depth.
Our new game is at the same time a hazard-free analog of the classical Boolean Karchmer-Wigderson game (Remark~\ref{rem:shortcondition}) and a generalization of the monotone Karchmer-Wigderson game to the set of all Boolean functions: it coincides with the monotone Karchmer-Wigderson game when played on monotone functions (Theorem~\ref{thm:coincide}).
In other words, the difference between the monotone Karchmer-Wigderson game and the Boolean Karchmer-Wigderson game is precisely the presence of hazards in the Boolean game.
We use this new definition to precisely determine the hazard-free formula size (Theorems~\ref{thm:mux-upper} and \ref{thm:mux-lower}) and the depth of hazard-free formulas of alternation depth\footnote{Alternation depth is one plus the maximum number of changes in the type of the gate in root-to-leaf paths.} $2$ (Theorem~\ref{thm:mux-dnf-depth}) of the multiplexer function $\MUX_n : \{0,1\}^{n+2^n} \to \{0,1\}$, which is a (non-monotone) Boolean function on $n+2^n$ input bits, defined via
\[
\MUX_n(s_1,\ldots,s_n,x_0,x_1,\ldots,x_{2^n-1}) =  x_{\bin(s_1,\ldots,s_n)},
\]
where $\bin(s_1,\ldots,s_n)$ is the natural number represented by the binary number $s_1s_2\cdots s_n$.

Our result breaks the monotone barrier, i.e., the hazard-derivatives of the multiplexer have lower complexity than the bound we prove.
To obtain matching upper and lower bounds on complexity we use the Karchmer-Wigderson game interpretation to give two new efficient hazard-free implementations of the multiplexer function: One is optimal for the formula size and one is optimal for the depth of hazard-free formulas of alternation depth $2$. 

In contrast to monotone complexity, which is mainly a theoretical concept, hazard-free complexity has applications in practice, not only in cybersecurity (\cite{TWMMCS:09}, \cite{HOITSMK:12}, \cite{BHTBKI:17}), but also for designing real-world circuits, for example when a distributed system of agents with unsynchronized clock domains performs a parallel computation, see \cite{FKLP:17, FFL:18, TFL:17, LM:16, BLM:17, BLM:18, BLM:20}. The hazard-free circuit depth (which is equal to the hazard-free formula depth) is a main parameter in this research area, directly correlated to a circuit's execution time.

An interesting incremental approach towards
proving super-polynomial formula size lower bounds for explicit functions, is to make progress by proving good lower bounds for formulas with more and more NOT gates \cite{F:75,BNT:96,BNT:98}. In Section~\ref{app:limited-hazard-free}, we show that instead of considering all implicants and implicates, we can choose any subset of implicants and implicates to obtain upper and lower bounds on \emph{limited} hazard-free formulas, formulas that are guaranteed to be hazard-free on some inputs but not others. That is, we can parameterize our game by the number of undefined inputs so that it interpolates between the hazard-free game and the general Boolean game. This gives us a natural way to make progress towards proving super-polynomial Boolean formula size lower bounds by proving super-polynomial lower bounds for more and more limited hazard-free formulas, until we prove a lower bound on formulas that may have hazards on any input. Limited hazard-free formulas are also of interest in practice, for example when it is known that the unstable bit can only appear in the position where two adjacent Gray code numbers differ \cite{FKLP:17, LM:16, BLM:17, BLM:18, BLM:20}. We are not aware of any applications of limited negation circuits for designing real-world circuits. 

The multiplexer function is also significant from the perspective of proving super-polynomial formula size lower bounds. Informally, it suffices to prove a lower bound for a composition of the multiplexer function with itself. For a formal statement, see \cite{EIRS:01, Meir:20}.

\subsection{Exact Bounds}

In Section~\ref{sec:exactsizeMUX} we determine the \emph{exact} hazard-free formula complexity of the multiplexer function. We achieve this by using a combination of an improvement in the upper bound (Huffman's \cite{Huf:57} construction gives only  $\size{\un}(\MUX_n) \leq 4^n + 2n3^{n-1}$) and an analysis of the hazard-free Karchmer-Wigderson game for the lower bound:
\[
\size{\un}(\MUX_n) = 2 \cdot 3^n - 1.
\]
It is known that there are De Morgan formulas (with hazards) of size $2^{n+1}(1+o(1))$ computing $\MUX_n$ \cite{LV:09},
i.e., $\size{}(\MUX_n) \leq 2^{n+1}\left(1 + \frac{1}{2n} + O(\frac{1}{n\log n}) \right)$. 
Our upper bound construction is a recursive application of the improved implementation of $\MUX_1$ in Figure~\ref{fig:mux1-eg}(c).
To prove the lower bound we reduce the Karchmer-Wigderson game for $\MUX_n$ from a communication game for the \emph{subcube intersection problem}.
Its communication matrix is highly structured, so that its rank can be determined and be used to find the lower bound.
The subcube intersection problem is the hazard-free generalization of the classical equality problem from communication complexity and could be of independent interest, especially for proving other hazard-free formula lower bounds.

Since all derivatives of $\MUX_n$ have monotone formulas of size at most $(n+1)2^n$ (Proposition~\ref{prop:derivatives}), the separation that we achieve breaks the monotone barrier.
Therefore, our lower bound is the first to separate the Boolean complexity and the hazard-free complexity of a function while breaking the monotone barrier\footnote{Note that breaking the monotone barrier can also be achieved using Khrapchenko's method for the parity function \cite{K:71},
which was interpreted as a Karchmer-Wigderson game in \cite{KW90},
but for the parity function the hazard-free complexity and the Boolean complexity coincide (every implementation of parity is automatically hazard-free): Parity requires $\Theta(n^2)$ formula size, but the derivatives of parity are all equal to the OR function, which requires $\Theta(n)$ formula size.
For the parity function the Boolean Karchmer-Wigderson game coincides 
with our hazard-free Karchmer-Wigderson game, so we obtain the same bounds.
}.
We consider this an important step forward towards establishing hazard-free computation as a new theoretical device that can serve as a true generalization of monotone circuit complexity.

Considering the depth (which is the same for circuits and formulas), we immediately obtain
$
\depth{\un}(\MUX_n) \geq \log_2(3) n \geq 1.58 n.
$
This lower bound separates the hazard-free circuit depth complexity and Boolean circuit depth complexity of $\MUX_n$, because 
$\depth{}(\MUX_n) \leq n + 3$ \cite{TZ:97, LV:11}. 
(In fact, $\depth{}(\MUX_n) = n+2$ for all $n\geq 20$ \cite{LV:11}.) 
Analogously to formula size, since all derivatives of $\MUX_n$ have monotone circuits of depth at most $n + \log_2(n) + 1$ (Proposition~\ref{prop:derivatives}), our separation breaks the monotone barrier.

In Section~\ref{sec:tworoundprot} we focus on the depth of hazard-free formulas of alternation depth $2$ for $\MUX_n$. These formulas are interesting in practice because certain programmable logic arrays produce implementations that have alternation depth $2$. We prove the exact complexity of the multiplexer function in this restricted model to be $2n+2$.
For the proof we exploit an old result by Huffman: the fact that in this restricted model Alice must communicate her prime implicant to Bob before Bob starts communicating. Therefore small-depth formulas can exist if and only if there are short prefix codes that allow Alice to communicate her prime implicant efficiently to Bob. Then using Kraft's inequality from information theory we show that there are prefix codes that achieve a depth upper bound of $2n+2$, but they cannot achieve $2n+1$.
One key idea is a distinction of cases between prime implicants of logarithmic size and prime implicants of super-logarithmic size.

For general hazard-free formula depth the upper bound of $2n+2$ is not optimal for $\MUX_n$, because we show in Theorem~\ref{thm:mux-depth-upper} that the depth is at most $2n+1$.
Note that this is significantly lower than the depth $3n$ achieved by the formula of optimal size in Theorem~\ref{thm:mux-upper}, and strictly lower than the depth that can be achieved by any formula of alternation depth $2$.
This gives a size-depth trade-off:
the size of this formula is only a factor of
$\frac{9}{8}$
more than the optimal size.
This construction is done recursively using the hazard-free Karchmer-Wigderson game.
It is crucial in this recursion that the induction hypothesis is \emph{not} the
monochromatic partitioning of the communication matrix of $\MUX_{n-1}$, but of an enlarged matrix that can be partitioned monochromatically using the same depth.

All upper bounds and lower bounds are proved using the framework of hazard-free Karchmer-Wigderson games. The lower bound proofs rely heavily on this framework. The game also played a crucial role in deriving the upper bounds given in Theorems~\ref{thm:mux-depth-upper} and \ref{thm:mux-dnf-depth}. The upper bound in Theorem~\ref{thm:mux-upper} can also be be proved without using the game (see Remark~\ref{rem:demystify}).

\subsection{Universal Upper Bounds}\label{subsec:universal}
One of the most fundamental and oldest questions in electronic circuit design is finding an upper bound on the size of circuits or formulas that holds for all Boolean functions \cite{S:49}.
For Boolean circuits and formulas, this question has been very satisfactorily answered. It is known that any $n$-bit Boolean function has circuits of size $(1+o(1))2^n/n$ \cite{L:58, Loz:96} and almost all Boolean functions require circuits of size $(1+o(1))2^n/n$ \cite{Lup:63}.
For Boolean formulas, the lower bound is $(1-o(1))2^n/\log(n)$ \cite{RS:42}, almost matched by the upper bound  $(1+o(1))\frac{2^n}{\log n}$ \cite{Lup:60, Loz:96}.

For hazard-free circuits, the situation is very similar to that of Boolean circuits: any $n$-bit Boolean function has a hazard-free circuit of size $O(2^{n}/n)$ (see, e.g., \cite[Section 7]{Jukna21}), thus matching Lupanov's upper bound \cite{L:58} up to constants.
Since hazard-free circuits are also Boolean circuits, the lower bound of $(1+o(1))2^n/n$ for almost all functions continues to hold for hazard-free circuits. 

For hazard-free formulas, this question is still open.
Huffman \cite{Huf:57} gives hazard-free implementations  for any function by representing it as a DNF where the set of terms is the set of all prime implicants of the function.
Since a function on $n$ variables may have as many as $\Omega(3^n/\sqrt{n})$ prime implicants \cite{CM:78} and each prime implicant may contain as many as $n$ literals, this translates into a worst-case bound of $O(\sqrt{n}\cdot 3^n)$ on the hazard-free formula complexity.

We make progress on this question by studying the multiplexer function.
In electronic circuit design, the multiplexer is often used as a programmable logic device. Indeed,
given any Boolean function $f : \Bool^n \mapsto \Bool$, we can implement it as: $f(x_1, \dotsc, x_n) = \MUX_n(x_1, \dotsc, x_n, f(0, 0, \dotsc, 0), \dotsc, f(1, 1, \dotsc, 1))$.
This implementation of $f$ is hazard-free if the implementation of $\MUX_n$ is hazard-free.
Therefore, any hazard-free formula upper bound for $\MUX_n$ gives an upper bound for the hazard-free formula complexity of \emph{all} $n$-bit Boolean functions.
Theorem~\ref{thm:mux-upper} gives such an improved upper bound of $2\cdot 3^n-1$ for the multiplexer function and hence our construction gives a new best worst-case hazard-free formula size implementation of size $2\cdot 3^n-1$, which was $O(\sqrt n \cdot 3^n)$ before.

Observe that in the world of Boolean circuits, Boolean formulas, and hazard-free circuits, the multiplexer upper bound is only a polynomial (in $n$) multiplicative factor away from the optimal bound.
We show in Theorem~\ref{thm:mux-lower} that our new bound is optimal for the multiplexer function.
This means that we cannot improve the universal upper bound further by directly using the multiplexer function. However, the best known lower bound for hazard-free formulas for $n$-bit functions is still the $2^n/\log(n)$ given by a counting argument. This creates an interesting situation that is different from the other three settings described in this section.
\begin{compactitem}
\item If there are $n$-bit functions such that the hazard-free formula size is asymptotically more than $2^n/\log(n)$, then a tight lower bound can be proved by only using some argument that exploits the \emph{semantic} property of hazard-freeness, such as the hazard-free Karchmer-Wigderson game we introduce in this paper. This is in contrast to the other settings where tight lower bounds can be obtained using a counting argument that only exploits the structure (or syntax) of the model.
\item Otherwise, all $n$-bit functions have hazard-free formulas that are smaller than the optimal hazard-free formula for the multiplexer function by a multiplicative factor that is exponential in $n$. This is also in stark contrast to the situation in the other three settings.
\end{compactitem}

\section{Preliminaries}

\paragraph{Formulas}
A \emph{Boolean formula} is a Boolean circuit whose graph is a tree.
That is, it is a formula over the De Morgan basis $\{\vee, \wedge, \neg\}$.
The $\vee$ and $\wedge$ gates have fan-in two and $\neg$ gates have fan-in one.
Using De Morgan's laws (which also work over the three-valued logic) the negations can be moved to the leaves: all
internal nodes are labeled with $\vee$ or $\wedge$ and all leaves are labeled with literals $x_i$ or $\neg x_i$. This is called a \emph{De Morgan formula}.
The \emph{size} of a De Morgan formula $F$, denoted $\size{}(F)$, is defined to be the number of leaves in it\footnote{
If all $\wedge$ and $\vee$ gates have fan-in two, then the number of leaves is always exactly one more than the number of gates (not counting negation gates) in $F$, which is a measure often used to describe circuit size.}.
The \emph{depth} of a formula $F$, denoted $\depth{}(F)$, is defined to be the length of the longest root-to-leaf path in $F$. 
For a Boolean function $f\colon\Bool^n\to\Bool$, we denote the minimal size 
of a De Morgan formula computing $f$ by $\size{}(f)$ and the minimal 
depth of a formula computing $f$ by $\depth{}(f)$. Similarly, 
in the hazard-free setting, let $\size{\un}(f)$ and $\depth{\un}(f)$ denote 
the minimal size and minimal depth of a hazard-free De Morgan formula computing $f$, respectively.
For a monotone function $f$ let $\size{+}(f)$ and $\depth{+}(f)$ denote 
the minimal size and minimal depth of a monotone formula computing $f$, respectively.
The \emph{alternation depth} of a formula is one plus maximum number of changes to the type of the gate in the sequence of gates in a root-to-leaf path.
For example, the alternation depth of the formula in Figure~\ref{fig:mux1-eg}(b) is $2$ and that of the formula in Figure~\ref{fig:mux1-eg}(c) is $3$.
We denote the minimal size and depth of hazard-free formulas of alternation depth $d$ using $\size{\un}_{d}(f)$ and $\depth{\un}_{d}(f)$, respectively.

\paragraph{Implicants and Implicates}
For a Boolean function $f:\{0,1\}^n\to\{0,1\}$
the preimage of a value $c\in\{0,1\}$ is denoted by $f^{-1}(c)$. For the hazard-free extension $\hfe{f}:\{0,\un,1\}^n\to\{0,\un,1\}$ the preimage of $\gamma\in\{0,\un,1\}$ is denoted by $\hfe{f}^{-1}(\gamma)$.
Elements $\alpha\in \hfe{f}^{-1}(1)$ are called \emph{implicants} of $f$.
A \emph{prime implicant} is an implicant in which no value from $\{0,1\}$ can be replaced by a $\un$ such that it is still an implicant, i.e., a prime implicant is an implicant that is minimal with respect to the \emph{instability partial order}, the partial order on $\{0, 1, \un\}$ defined by $\un \leq 0$ and $\un \leq 1$.
Elements $\alpha\in \hfe{f}^{-1}(0)$ are called \emph{implicates} of $f$.
A \emph{prime implicate} is an implicate in which no value from $\{0,1\}$ can be replaced by a $\un$ such that it is still an implicate, i.e., a prime implicate is an implicate that is minimal with respect to the instability partial order.
We occasionally identify an implicant $\alpha$ with the Boolean function that is 1 exactly on the hypercube of resolutions of $\alpha$, and an implicate $\beta$ with the Boolean function that is 0 exactly on the hypercube of resolutions of $\beta$.

\paragraph{Communication} We assume familiarity with the basic definitions of
communication complexity (see, e.g., \cite{KN96, RY20}). 
Let $K\colon{A}\times{B}\to 2^O$ be a function that maps tuples to nonempty subsets of a set $O$.
For the purposes of this paper we will only be interested in 
deterministic communication complexity where Alice gets 
$\alpha \in {A}$, Bob gets $\beta\in{B}$ and their goal is to 
determine some value in $K(\alpha,\beta)$ while minimizing the communication (number of bits exchanged). Let $\Pi$ be a deterministic communication protocol solving $K$. 
Then the \emph{communication cost} of $\Pi$, denoted $\cc(\Pi)$, 
is defined to be the maximum number of bits exchanged 
on any pair of inputs $(\alpha,\beta)$ when following $\Pi$. 
Let $\cc(K)$ denote the minimum cost over all protocols solving $K$.
Recall that the leaves of a protocol induce a partition of
${A}\times{B}$ into combinatorial rectangles. 
We denote the number of such combinatorial rectangles in a protocol $\Pi$ by $\monorect(\Pi)$ 
and the minimum number of leaves in a protocol solving $K$ by $\monorect(K)$.

We will often work with the communication matrix $M_K$ of dimensions 
$|{A}| \times |{B}|$ associated with a function $K$.  
The rows and columns of $M_K$ are indexed by the elements of ${A}$ 
and ${B}$, respectively. The $(\alpha,\beta)$-th entry of $M_K$ is defined to 
be $K(\alpha,\beta)$. The leaves of a protocol $\Pi$ solving $K$ 
partitions the communication matrix $M_K$ 
into $\monorect(\Pi)$ many monochromatic combinatorial rectangles, where a combinatorial rectangle $A'\times B'$ ($A'\subseteq A$, $B'\subseteq B$) is called monochromatic if there exists $o \in O$ with $\forall (\alpha,\beta)\in A'\times B' : o \in K(\alpha,\beta)$.
We will often use $K$ and $M_K$ interchangeably.

\section{Hazard-Derivatives and the Monotone Barrier}

For $x,y \in\Bool^n$, we define $x\oplus\un\cdot y$ to be the string $\alpha\in\Tri^n$ such that for all $i\in[n]$, $\alpha_i = x_i$ if $y_i=0$, and otherwise $\alpha_i=\un$. 
Let $f$ be a Boolean function on $n$ variables. Its \emph{hazard derivative} is a Boolean function on $2n$ variables denoted $\der f(x; y)$ that evaluates to $1$ if and only if $\hfe{f}(x \oplus \un \cdot y) = \un$, i.e., there are two resolutions of $x \oplus \un \cdot y$, say $a$ and $b$, such that $f(a) = 0$ and $f(b) = 1$.
In other words, $\der f(x; y)=1$ if and only if the function $f$ is \emph{not} constant on the subcube of all resolutions of $x \oplus \un \cdot y$. 
For example consider the multiplexer function $\MUX_n(s,x)$ which is defined as 
\[
\MUX_n(s_1,\ldots,s_n,x_0,x_1,\ldots,x_{2^n-1}) =  x_{\bin(s_1,\ldots,s_n)},
\]
where $\bin(s_1,\ldots,s_n)$ is the natural number represented by the binary number $s_1s_2\cdots s_n$. 
Its hazard derivative $\der\MUX_n(s,x;t,y) = 1$ if and only if the string $x\oplus\un\cdot y$ restricted to the positions given by the subcube of all resolutions of $s\oplus\un\cdot t$ is neither the all zeroes nor the all ones string. Equivalently, we can say that either $x$ restricted to the positions given by the subcube of all resolutions of $s\oplus\un\cdot t$ is not all zeroes or all ones, or $y$ restricted to the positions given by the subcube of all resolutions of $s\oplus\un\cdot t$ is not all zeroes. 

Notice that for a fixed value $a\in \{0,1\}^n$,  the function $\der f(a; y):\{0,1\}^n\to\{0,1\}$, is a monotone function, see \cite[Lemma~4.6]{IKLLMS19}. 
The key observation that connects hazard-free circuits to monotone circuits is that given a hazard-free circuit $C$ for $f$ and any Boolean string $a$, we can construct a monotone circuit for $\der f(a; y)$ that is no larger in size than $C$, see \cite[Thm.~4.9]{IKLLMS19}. Therefore, in order to prove lower-bounds for hazard-free circuits for $f$, one only needs to identify a Boolean string $a$ such that $\der f(a; y)$ is a hard function for monotone circuits, i.e., has high monotone circuit complexity. This allows us to transfer a wealth of known monotone circuit lower bounds to the hazard-free world. This works analogously for hazard-free De Morgan formulas (in this case the hazard-derivative can be constructed as a monotone De Morgan formula instead of a monotone circuit).

The best known construction for hazard-free De Morgan formulas for $\MUX_n$ has size $2 \cdot 3^n - 1$ (See Theorem~\ref{thm:mux-upper}). Can we use derivatives to prove that this is optimal? No. We show that \emph{all} derivatives of $\MUX_n$ have monotone De Morgan formulas of size at most $(n+1)2^n$. This is an instance of the \emph{monotone barrier}.

\begin{proposition}\label{prop:derivatives}
Fix any $(s,x) \in \{0,1\}^{n+2^n}$.
The function $\der\MUX_n(s,x;t,y):\{0,1\}^{n+2^n}\to\{0,1\}$ has monotone De Morgan formulas of size at most $(n+1)2^n$.
\end{proposition}
\begin{proof}
Fix $(s,x)\in\Bool^{n+2^n}$. Recall $\der\MUX_n(s,x;t,y)=1$ if and only if $x$ restricted to the positions given by the subcube of all resolutions of $s\oplus\un\cdot t$ is not all zeroes or all ones, or $y$ restricted to the positions given by the subcube of all resolutions of $s\oplus\un\cdot t$ is not all zeroes.


For the sake of clarity we first give a monotone implementation of 
$\der\MUX_n$ when $(s,x)$ is fixed to the all zeroes input. 
When $(s,x)=(0,0)$ (i.e., the all zeroes string), then $x$ restricted to the positions given by the subcube of all resolutions of $s\oplus\un\cdot t$ is all zeroes. Therefore, $\der\MUX_n(0,0;t,y) = 1$ if and only if $y$ restricted to the positions given by the subcube of all resolutions of $s\oplus\un\cdot t$ is not the all zeroes string. 
Furthermore, since $s=0$, the subcube of all resolutions of $s\oplus\un\cdot t$ is the set of points $b\in\Bool^n$ such that $b\leq t$. 
Thus, $\der\MUX_n(0,0;t,y) = \bigvee_{b \in \Bool^n}y_{\bin(b)} \wedge(\bigwedge_{i : b_i =1}t_i)$.

For the general case when $(s,x)\in\Bool^{n+2^n}$ is arbitrary (but fixed), we need to also verify whether $x$ restricted to the positions given by the subcube of all resolutions of $s\oplus\un\cdot t$ is neither all zeroes nor all ones. Equivalently, for $b\in\Bool^n$, when $x_{\bin(b)} \neq \MUX_n(s,x) = x_{\bin(s)}$, we need not check if $y_{\bin(b)} = 1$. Thus we have the following monotone implementation
\[\der\MUX_n(s,x;t,y) = \left(\bigvee_{\substack{b\in\Bool^n\colon \\ x_{\bin(b)} = x_{\bin(s)}}}y_{\bin(b)}\wedge\left(\bigwedge_{i\colon b_i\neq s_i}t_i\right)\right) \vee \left(\bigvee_{\substack{b\in\Bool^n\colon\\ x_{\bin(b)}\neq x_{\bin(s)}}}\bigwedge_{i\colon b_i\neq s_i}t_i\right).\]
The size bound easily follows. 
\end{proof}

The above proposition shows that the derivative method cannot yield a lower bound bigger than $(n+1)2^n$ for hazard-free De Morgan formulas for $\MUX_n$. We now proceed to develop a framework that will allow us to prove that $2 \cdot 3^n - 1$ is the optimal size for $\MUX_n$. This is the first result that proves a hazard-free circuit lower bound without relying on an existing monotone circuit lower bound, i.e., that breaks the monotone barrier.

\section{A Karchmer-Wigderson Game for Hazard-free Computation}

In this section we give a natural generalization of the classical Karchmer-Wigderson game,
which captures the complexity of \emph{hazard-free} computation.
We begin with recalling the framework of Karchmer-Wigderson games \cite{KW90}.

\begin{definition}[\cite{KW90}]
  \label{def:kw-rel}
  Let $f\colon \Bool^n \to \Bool$ be a Boolean function.
  The \emph{Karchmer-Wigderson game} of $f$, denoted $\kw{}_f$,
  is the following communication problem: 
  Alice gets $a \in \{0,1\}^n$ with $f(a)=1$, 
  Bob gets $b \in \{0,1\}^n$ with $f(b)=0$ and
  their goal is to determine a coordinate $i \in [n]$ such that $a_i \neq b_i$.
\end{definition}
They also gave the following monotone version of the game. 
\begin{definition}[\cite{KW90}]
  \label{def:kw-relmon}
  Let $f\colon \Bool^n \to \Bool$ be a monotone Boolean function.
  The \emph{monotone Karchmer-Wigderson game} of $f$, denoted $\kw{+}_f$,
  is the following communication problem:
  Alice gets $a \in \{0,1\}^n$ with $f(a)=1$, Bob gets $b \in \{0,1\}^n$ with $f(b)=0$ and
  their goal is to determine a coordinate $i \in [n]$ such that $1=a_i \neq b_i=0$.
\end{definition}
The seminal work of Karchmer and Wigderson \cite{KW90} showed that
the communication complexity of the $\kw{}_f$ game (resp., $\kw{+}_f$ game)
characterizes the size and depth complexity of De Morgan formulas 
(resp., monotone formulas).

\begin{theorem}[\cite{KW90}]
  \label{thm:kwthm}
Let $f\colon\Bool^n\to\Bool$ be a Boolean function. Then, 
\[\depth{}(f) = \cc(\kw{}_f),\quad \mbox{ and }\quad \size{}(f)= \monorect(\kw{}_f).\]
Let $f\colon\Bool^n\to\Bool$ be a monotone Boolean function. Then, 
\[\depth{+}(f) = \cc(\kw{+}_f),\quad \mbox{ and }\quad \size{+}(f)= \monorect(\kw{+}_f).\]
\end{theorem}

We now extend the Karchmer-Wigderson games to the hazard-free setting.
For a Boolean function $f\colon\Bool^n\to\Bool$, recall from \eqref{eq:defhazfreeextension} that 
$\hfe{f}\colon\Tri^n\to\Tri$ is the hazard-free extension of $f$. 

\begin{definition}[Hazard-free Karchmer-Wigderson game]
  \label{def:hf-kw-rel}
  Let $f\colon \Bool^n \to \Bool$ be a Boolean function.
  The \emph{hazard-free Karchmer-Wigderson game} of $f$, denoted $\kw{\un}_f$,
  is the following communication problem:
  Alice gets $\alpha \in \Tri^n$ with $\hfe{f}(\alpha)=1$, 
  Bob gets $\beta \in \Tri^n$ with $\hfe{f}(\beta)=0$ and
  their goal is to determine a coordinate $i \in [n]$ such that 
  $\alpha_i \neq \beta_i$ and furthermore $\alpha_i \neq \un$ 
  and $\beta_i \neq \un$.
 \end{definition}
 
An example of a communication matrix for this game is shown in Figure~\ref{fig:firstmuxmatrix}.
\begin{figure}
\begin{subfigure}{0.5\textwidth}
\centering
\[
\begin{pNiceArray}{ccccccc}[first-row,first-col]
~   & 000 & 00\un & 001 & \un00 & 100 & 1\un0 & 110 \\
010 & 2 & 2 & 2,3 & 2 & 1,2 & 1 & 1 \\
01\un & 2& 2& 2 & 2 & 1,2 & 1 & 1 \\
011 & 2,3& 2& 2 & 2,3 & 1,2,3 & 1,3 & 1,3 \\
\un11 & 2,3 & 2 & 2 & 2,3 & 2,3 & 3 & 3\\
 101 & 1,3 & 1 & 1 & 3 & 3 & 3 & 2,3 \\
 1\un 1& 1,3 & 1 & 1 & 3 & 3 & 3 & 3 \\
111 & 1,2,3 & 1,2 & 1,2 & 2,3 & 2,3 & 3 & 3 
\end{pNiceArray}
\]
    \caption{The communication matrix $M_{\kw{\un}_{\MUX_1}}$.\vspace{2\baselineskip}}
\end{subfigure}
\begin{subfigure}{0.5\textwidth}
\[
\begin{pNiceArray}{ccccccc}[first-row,first-col,margin,rules/width=3pt]
~   & 000 & 00\un & 001 & \un00 & 100 & 1\un0 & 110 \\
010 & \Block[fill=red!15]{4-4}{} 2 & 2 & 2 & 2 & \Block[fill=yellow!30]{3-3}{} 1 & 1 & 1 \\
01\un & 2& 2& 2 & 2 & 1 & 1 & 1 \\
011 & 2& 2& 2 & 2 & 1 & 1 & 1 \\
\un11 & 2 & 2 & 2 & 2& \Block[fill=blue!15]{1-3}{} 3& 3 & 3\\
 101 & \Block[fill=orange]{3-3}{} 1 & 1 & 1 & \Block[fill=green!15]{3-4}{} 3 & 3 & 3 & 3 \\
 1\un 1& 1 & 1 & 1 & 3 & 3 & 3 & 3 \\
111 & 1 & 1 & 1 & 3 & 3 & 3 & 3 
\end{pNiceArray}
\]
    \caption{The monochromatic partition of $M_{\kw{\un}_{\MUX_1}}$.
    It is a result of applying the construction from Lemma~\ref{lem:formula-protocol} to Figure~\ref{fig:mux1-eg}(c).}
\end{subfigure}
    \caption{The communication matrix $M_{\kw{\un}_{\MUX_1}}$ and a monochromatic partition.}
    \label{fig:firstmuxmatrix}
\end{figure}
Observe that, for all $(\alpha,\beta)\in \hfe{f}^{-1}(1) \times \hfe{f}^{-1}(0)$, 
there exists an $i \in [n]$ such that 
$\alpha_i \neq \beta_i$, $\alpha_i \neq \un$ and $\beta_i \neq \un$~\footnote{Assume the contrary. Then $\alpha$ and $\beta$ must have some resolution in common, which we call $a$.
Hence $f(a)=1$ and $f(a)=0$, which is a contradiction.},
which implies that all cells in the communication matrix are nonempty.  

\begin{remark}\label{rem:shortcondition}
Note that the wordy condition ``$\alpha_i \neq \beta_i$ and $\alpha_i \neq \un$ and $\beta_i \neq \un$'' is equivalent to the simple $\alpha_i \oplus \beta_i = 1$, which is in complete analogy to $a_i \oplus b_i = 1$ in the classical Karchmer-Wigderson game, see Def.~\ref{def:kw-rel}, whereas we will show in Theorem~\ref{thm:coincide} that our game is actually a generalization of the monotone Karchmer-Wigderson game to the domain of all Boolean functions.
\end{remark}

Now using this generalized game $\kw{\un}_f$ we characterize 
the complexity of hazard-free De Morgan formulas for Boolean functions. 

\begin{theorem}
\label{thm:hf-kwthm}
Let $f\colon\Bool^n\to\Bool$ be a Boolean function. Then, 
\[\depth{\un}(f) = \cc(\kw{\un}_f),\quad \mbox{ and }\quad \size{\un}(f)= \monorect(\kw{\un}_f).\]
\end{theorem}
The proof is a natural generalization of the proof of Theorem~\ref{thm:kwthm} and is provided in Section~\ref{sec:hazfreeKWgameproofs}.

\begin{remark}\label{rem:relatedwork}
We remark that a variant of the game $\kw{\un}_f$ has been
considered in prior works \cite{Hastad98,FMT21}.  
In this variant, the inputs to Alice and Bob remains the same 
but the goal is \emph{different}. 
More formally, Alice gets $\alpha\in \hfe{f}^{-1}(1)$,  
Bob gets $\beta\in\hfe{f}^{-1}(0)$ and their goal is to determine 
a coordinate $i\in[n]$ such that $\alpha_i\neq \beta_i$. 
That is, now a coordinate where one of them has $\un$ 
and the other has $0$ or $1$ is a valid answer. 
This is the subtle but crucial difference with respect to our game (Definition~\ref{def:hf-kw-rel}), where we forbid such answers 
by requiring that $\alpha_i\neq\un$ and $\beta_i\neq\un$. 
\end{remark}

\begin{remark}
Building on \cite{Raz:95} and \cite{Pudlak:10}, \cite{sokolov2017dag}
generalizes the Karchmer-Wigderson result from formulas to De Morgan circuits (i.e., Boolean circuits with negations only at the inputs). The corresponding communication games are played on directed acyclic graphs.
We remark that their proofs also work for hazard-free complexity, so we get a tight correspondence between hazard-free De Morgan circuit size and the correspoding hazard-free $\kw{}$-game.
\end{remark}

\subsection{Restriction to prime implicants and prime implicates}
We now prove that we can
restrict our attention to small (in some cases significantly smaller) submatrices of the communication matrix.
We will use this restricted version of the game to show that for monotone functions the hazard-free Karchmer-Wigderson game is \emph{equivalent} to the monotone Karchmer-Wigderson game, see Theorem~\ref{thm:coincide}.

\begin{theorem}\label{thm:kw-prime}
For any function $f$, the complexity (works for size and also for depth) of the game $\kw{\un}_f$ remains unchanged even if we restrict  Alice's input to prime implicants and Bob's input to prime implicates.
\end{theorem}
\begin{proof}
The complexity of the restricted game is obviously at most the complexity of the original game, since the game is now being played on a submatrix of the original matrix. For the other direction, observe that given an arbitrary implicant $\alpha$ and an arbitrary implicate $\beta$, Alice can choose a prime implicant $\alpha'$ that is obtained by flipping some stable bits in $\alpha$ to $\un$ and Bob can choose a prime implicate $\beta'$ that is obtained by flipping some stable bits in $\beta$ to $\un$, and play the restricted game on the input $(\alpha', \beta')$. Any valid answer in the restricted game is also a valid answer in the original game, since we are only flipping stable bits to $\un$. This proves that the complexity of the original game is at most the complexity of the restricted game. Therefore, both games have the same complexity.
\end{proof}

For example, consider the communication matrix of $\kw{\un}_{\MUX_1}$ given below, where we restricted the rows and columns to prime implicants and prime implicates. It is a submatrix of Figure~\ref{fig:firstmuxmatrix}(a).
\begin{equation}\label{eq:reducedmuxmatrix}
\begin{pNiceMatrix}[first-row,first-col]
~       & 00{\un}          &   {\un}00         & 1{\un}0 \\
01{\un} & \entry{2}      & \entry{2}       & \entry{1}   \\
{\un}11 & \entry{2}      & \entry{2, 3}  & \entry{3} \\
1{\un}1 & \entry{1}        & \entry{3}       & \entry{3} 
\end{pNiceMatrix}
.\end{equation}
We will use this equivalent reduced form of the hazard-free Karchmer-Wigderson game in the rest of the paper.

There is a natural counterpart to Theorem~\ref{thm:kw-prime} in the monotone world: in Definition~\ref{def:kw-relmon}, we can assume without loss of generality that Alice's input has minimal number of ones and Bob's input has maximal number of ones. We show that we can view the hazard-free Karchmer-Wigderson game as a generalization of the monotone Karchmer-Wigderson game.

\begin{theorem}\label{thm:coincide}
Let $f\colon\Bool^n\to\Bool$ be a \emph{monotone} function. 
Then, the games $\kw{\un}_f$ and $\kw{+}_f$ are equivalent.
\end{theorem}
\begin{proof}
First, we show that the complexity of $\kw{\un}_f$ is at most that of $\kw{+}_f$. Using Theorem~\ref{thm:kw-prime}, we can assume that Alice's input is a prime implicant and Bob's input is a prime implicate. Since $f$ is monotone, any prime implicant of $f$ contains only $1$s and $\un$'s. Similarly, any prime implicate of $f$ contains only $0$s and $\un$'s. Now, Alice can flip every $\un$ in her input to $0$ and Bob can flip every $\un$ in his input to $1$ and play the game $\kw{+}_f$. Notice that by Definition~\ref{def:kw-relmon}, the output of this game will be a position where Alice's and Bob's input had different stable values originally.

For the other direction, Alice can flip every $0$ in her input to $\un$ and Bob can flip every $1$ in his input to $\un$. Since $f$ is monotone, Alice still has an input in $\hfe{f}^{-1}(1)$ and Bob still has an input in $\hfe{f}^{-1}(0)$. Now, the output of the game $\kw{\un}_f$ on these new inputs will also be a valid output for the game $\kw{+}_f$ since all stable bits in Alice's input are $1$ and all stable bits in Bob's input are $0$.
\end{proof}

\begin{remark}
We note another perspective on hazard-free $\kw{}$-game through the lens of monotone $\kw{}$-game. For this purpose we need to associate sets with implicants and implicates. For an implicant 
$\alpha \in \hfe{f}^{-1}(1)$ we define the set $S_\alpha$ of literals as follows: if $\alpha_i = 1$, then $x_i \in S_\alpha$ and if $\alpha_i=0$, then $\neg x_i\in S_\alpha$. In particular, if $\alpha_i=\un$ then neither $x_i$ nor $\neg x_i$ belongs to $S_\alpha$. Furthermore, $\alpha$ is said to be a \emph{prime} implicant if $S_\alpha$ is \emph{minimal} with respect to set inclusion. That is, no strict subset of $S_\alpha$ is also an implicant. 
Analogously, for an implicate $\beta \in \hfe{f}^{-1}(0)$ we define the set $T_\beta$ of literals as follows: if $\beta_i=1$, then $\neg x_i \in T_\beta$ and if $\beta_i=0$, then $x_i\in T_\beta$. Furthermore, $\beta$ is said to be a \emph{prime} implicate if $T_\beta$ is \emph{minimal} with respect to set inclusion.

By using Theorem~\ref{thm:kw-prime} and its counterpart in the monotone world, we observe that \emph{both} the monotone $\kw{}$-game and the hazard-free $\kw{}$-game has the exact same definition: Alice gets a set of literals corresponding to a prime implicant and Bob gets a set of literals corresponding to a prime implicate. Their goal is to find a literal in the intersection of the two sets.

This means that playing (this reduced version of) the hazard-free $\kw{}$-game is the same as playing the monotone $\kw{}$-game on a function that is not necessarily monotone.
\end{remark}
\begin{remark}
It is instructive to gain yet another perspective on hazard-free $\kw{}$-games via an application of a well-known result in communication complexity that informally says, for every communication total search problem $S$ (in particular also for the hazard-free $\kw{}$-game) there exists a partial monotone function $g$ such that 
the monotone $\kw{}$-game for $g$ is equivalent to the communication problem $S$. For a formal statement see~\cite[Lemma~2.3]{Gal:01} or \cite[Proposition~2.10]{PR:18} (and references therein). In our setting the partial monotone function $g:\Bool^{2n} \to \{0,1,\ast \}$ is defined as follows: an input $(y_1,\ldots ,y_n,z_1,\ldots ,z_n)$ is said to be admissible if for all $i\in[n]$ we have $y_i \cdot z_i = 0$, otherwise we call the input inadmissible. On an inadmissible input $g$ is undefined, which is denoted by $\ast$. Given an admissible input $(y,z)$, let us define 
$\alpha \in \Tri^n$ such that for all $i\in[n]$, $\alpha_i=y_i$ if $y_i\oplus z_i=1$, and otherwise $\alpha_i=\un$. Then, $g(y,z) = \Psi(\hfe{f}(\alpha))$, where $\Psi(0)=0$, $\Psi(1)=1$, and $\Psi(\un)=\ast$.  

Now it can be shown (see, e.g., \cite[Proposition~2.10]{PR:18}) that the hazard-free $\kw{}$-game for the Boolean function $f$ is equivalent to the monotone $\kw{}$-game for the partial function $g$. Thus the monotone $\kw{}$-games \emph{on partial functions} are expressive enough to capture hazard-free $\kw{}$-games.
Note that $g$ has twice as many inputs as~$f$.
\end{remark}

\section{Hazard-Free Formulas for the Multiplexer Function}
\label{sec:exactsizeMUX}
We now use the hazard-free Karchmer-Wigderson game to give 
improved constructions of hazard-free formulas as well as proofs of their optimality.
Our starting point is the observation that the commonly used hazard-free formula for 
$\MUX_1$ in Figure~\ref{fig:mux1-eg}(b) is \emph{not} optimal w.r.t.\ size.
We find an optimal formula for it (Figure~\ref{fig:mux_formula}) which in turn leads to an optimal formula 
of size $2\cdot 3^n -1$ for $\MUX_n$. Following the discussion in Subsection~\ref{subsec:universal}, this upper bound also applies to all $n$-bit Boolean functions and improves upon Huffman's construction \cite{Huf:57}. 
However, a gap between the upper and lower bound still remains. We begin with some necessary basics on 
the multiplexer function.

\subsection{The Multiplexer Function and its Communication Matrix}
Recall, the multiplexer function $\MUX_n\colon\Bool^{n+2^n}\to\Bool$ is 
a Boolean function on $n+2^n$ variables defined as 
\begin{align}\label{eq:defn-mux}
    \MUX_n(s_1,\ldots,s_n,x_0,x_1,\ldots,x_{2^n-1}) =  x_{\bin(s_1,\ldots,s_n)},
\end{align}
where $\bin(s_1,\ldots,s_n)$ is the natural number represented by the binary number $s_1s_2\cdots s_n$.
We will be studying the communication matrix $M_{\kw{\un}_{\MUX_n}}$ of the hazard-free game $\kw{\un}_{\MUX_n}$. Following Theorem~\ref{thm:kw-prime}, 
we will restrict our attention to the submatrix given by 
prime implicants and implicates of $\MUX_n$. 
The following proposition gives the structure of the prime implicants and prime implicates.
\begin{proposition}\label{prop:mux-prime}
For any $n \geq 1$ and any string $\alpha \in \Tri^n$, there exist unique strings $\alpha' \in \{\un, 1\}^{2^n}$ and $\beta' \in \{0, \un\}^{2^n}$ such that $\alpha\alpha' \in \Tri^{n+2^n}$ is a prime implicant of $\MUX_n$ and $\alpha\beta' \in\Tri^{n+2^n}$ is a prime implicate of $\MUX_n$.
\end{proposition}
\begin{proof}
For $\alpha \in \Tri^n$, consider the string $\alpha' \in \{\un, 1\}^{2^n}$ that has $1$s at positions indexed by the resolutions of $\alpha$ and that has $\un$'s elsewhere. We have $\hfe{\MUX}_n(\alpha\alpha') = 1$ showing that $\alpha\alpha'$ is an implicant. We now show it is a prime implicant. If any $1$s in $\alpha'$ are made a $\un$, then the output becomes a $\un$, because a resolution of $\alpha$ now indexes into a $\un$. If any Boolean value in $\alpha$ is made a $\un$, then at least one resolution of the selector bits is a position in the data bits that is a $\un$. Therefore this implicant is minimal. This is also that only prime implicant that can be obtained by extending $\alpha$ because the $\alpha'$ part is minimal and all $1$s in it are necessary. The argument for prime implicates is symmetric.
\end{proof}

The following proposition states the inductive structure of communication matrices of $\MUX_n$.
\begin{proposition}
\label{prop:mux-comm-matrix}
The communication matrix of $\kw{\un}_{\MUX_n}$, when restricted to prime implicants and prime implicates, has the following inductive structure:
\begin{compactitem}
    \item For $n = 1$, 
    \begin{align*}
        M_{\kw{\un}_{\MUX_1}} = \; \begin{pNiceMatrix}[first-row,first-col]
    ~       & 00{\un}          &   {\un}00         & 1{\un}0 \\
    01{\un} & \entry{x_0}      & \entry{x_0}       & \entry{s}   \\
    {\un}11 & \entry{x_0}      & \entry{x_0, x_1}  & \entry{x_1} \\
    1{\un}1 & \entry{s}        & \entry{x_1}       & \entry{x_1} 
    \end{pNiceMatrix}
    .
    \end{align*}
    \item   For $n\geq2$, 
    \begin{align*}
    M_{\kw{\un}_{\MUX_n}} = \begin{pNiceMatrix}[first-row,first-col,margin]
    ~ & 0   &   {\un}         & 1 \\
    \phantom{00}0 & M_0 & M_0 & s_1 \\
    \phantom{00}{\un} & M_0 & M_0 \cup M_1 & M_1   \\
    \phantom{00}1 & s_1 & M_1  & M_1 
    \end{pNiceMatrix},
    \end{align*}
    where the row (resp., column) labeled $\gamma \in \Tri$ represents the set of prime implicants (resp., prime implicates) with $s_1 = \gamma$. We define the formulas
    \begin{align*}
     F_0 &=\MUX_{n-1}(s_2\dotsc, s_{n-1}, x_0, \dotsc, x_{2^{n-1}-1}) = \MUX_n(0,s_2,\ldots,s_n,x_0,x_1,\ldots,x_{2^n-1}),\\ 
     F_1 &=\MUX_{n-1}(s_2, \dotsc, s_{n}, x_{2^{n-1}}, \dotsc, x_{2^{n}-1}) =  \MUX_n(1,s_2,\ldots,s_n,x_0,x_1,\ldots,x_{2^n-1}),
     \end{align*}
     and matrices $M_0 := M_{\kw{\un}_{F_0}}$, $M_1 := M_{\kw{\un}_{F_1}}$, $s_1$ stands for a block matrix of all entries $s_1$, and 
    $M_0 \cup M_1$ is obtained by taking entry-wise union of $M_0$ and $M_1$. In other words, $M_0 \cup M_1$ represents the matrix where the $(i,j)$ entry equals $(M_{0})_{i,j}\cup (M_{1})_{i,j}$.
\end{compactitem}
Note that in the communication matrix we have changed the entries from indices of variables to their labels, as in 
instead of $1,2,3$ we write the more intuitive symbols $s_1,\dotsc ,s_n,x_0,x_1,\dotsc,x_{2^n-1}$ for better readability (cp.~\eqref{eq:reducedmuxmatrix}).

\end{proposition}
See Figure~\ref{fig:commmatrix} for the example of $M_{\kw{\un}_{\MUX_2}}$.
\begin{figure}
    \centering
\[
\begin{pNiceMatrix}[margin,first-row,first-col]
{} & 000\un\un\un & 0\un00\un\un & 01{\un}0{\un}{\un} & \un 00{\un}0{\un} & \un\un 0000 & \un 1{\un}0{\un}0 & 10{\un}{\un}0{\un} & 1\un{\un}{\un}00 & 11{\un}{\un}{\un}0 \\
001{\un}{\un}{\un} & x_{0} & x_{0} & s_2 & x_{0} & x_{0} & s_2 & s_1 & s_1 & s_1, s_2 \\
0{\un}11\un\un & x_{0} & x_{0}, x_{1} & x_{1} & x_{0} & x_{0},x_{1} & x_{1} & s_1 & s_1 & s_1 \\
01\un1\un\un & s_2 & x_{1} & x_{1} & s_2 &x_{1}& x_{1} & s_1, s_2 & s_1 & s_1 \\
\un 01\un1\un& x_{0} & x_{0} & s_2 & x_{0}, x_{2} &x_{0},x_{2}&  s_2 & x_{2} & x_{2} & s_2 \\
\un\un 1111&  x_{0} & x_{0},x_{1} & x_{1} & x_{0},x_{2}&x_{0},x_{1},x_{2},x_{3}& x_{1},x_{3} & x_{2} & x_{2},x_{3}& x_{3}\\
\un 1\un1\un1 & s_2 & x_{1} & x_{1} & s_2 &x_{1},x_{3}& x_{1}, x_{3} & s_2 & x_{3} & x_{3} \\
10\un\un1\un & s_1 & s_1 & s_1, s_2 & x_{2} &x_{2}& s_2 & x_{2} & x_{2} & s_2 \\
1\un\un\un11 & s_1 & s_1 & s_1 & x_{2} &x_{2},x_{3}& x_{3} & x_{2} & x_{2}, x_{3} & x_{3} \\
11\un\un\un1 & s_1, s_2 & s_1 & s_1 & s_2 &x_{3}& x_{3} & s_2 & x_{3} & x_{3} 
\end{pNiceMatrix}
.\]
    \caption{The communication matrix for $\kw{\un}_{\MUX_2}$.}
    \label{fig:commmatrix}
\end{figure}
\begin{proof}[Proof of Proposition~\ref{prop:mux-comm-matrix}]
For $n=1$ the proof follows by inspection and when $n\geq 2$  it follows from the following recursive decomposition:  
$\MUX_n = \MUX_1(s_1, F_0, F_1)$. 
\end{proof}

We also need the following well-known general technique used in communication complexity that allows us to exploit 
repeated submatrices within a communication matrix.
\begin{proposition}\label{prop:parallel}
Let $M$ be a communication matrix such that 
$M = \begin{pmatrix}A& A\\A&A\end{pmatrix}$ or 
$M=\begin{pmatrix}A& A\end{pmatrix}$. Then, 
\[\cc(M) = \cc(A)\quad \mbox{ and }\quad \monorect(M) = \monorect(A).\]
\end{proposition}
\begin{proof}
Since $A$ is a submatrix of $M$, the following inequalities are self-evident
\[\cc(M) \geq \cc(A)\quad \mbox{ and }\quad \monorect(M) \geq \monorect(A).\]
For the other direction
we treat only the case $M = \begin{pmatrix}A& A\\A&A\end{pmatrix}$, as the other case follows because $\begin{pmatrix}A& A\end{pmatrix}$ is a submatrix of 
$\begin{pmatrix}A& A\\A&A\end{pmatrix}$.
We consider a protocol $\Pi$ of $A$. 
Using $\Pi$ we give a protocol $\Pi'$ for $M$ such that 
$\cc(\Pi')=\cc(\Pi)$ and $\monorect(\Pi') =\monorect(\Pi)$.

Without loss of generality, assume Alice is the first player to start in $\Pi$ and she sends a bit to indicate whether her input lies in the set of rows $R_1$ or $R_2$ such that the disjoint union $R_1 \uplus R_2$ is the set of all rows in $A$. Now let $R_1'$ be the union of rows $R_1$ from each copy of $A$ within $M$ and $R_2'$ be the union of rows $R_2$ from each copy of $A$ within $M$. Clearly, $R_1'\uplus R_2'$ is the set of all rows in~$M$. 
Then in $\Pi'$ too, Alice will start by sending a bit to indicate whether her input lies in $R_1'$ or $R_2'$. 
Upon receiving the message from Alice,
Bob now communicates using
completely analogous adjustments to the protocol
$\Pi$. 
The two players proceed in this way
and keep making these adjustments
until they reach the end of $\Pi$. 
From the protocol it follows that at the end of $\Pi'$ each rectangle
in $M$
is a union of the same rectangles from each copy of $A$, and thus monochromatic. Clearly, we also have $\cc(\Pi') = \cc(\Pi)$ and 
$\monorect(\Pi') = \monorect(\Pi)$. 
Since $\Pi$ is an arbitrary protocol for $A$, we obtain $\cc(M) \leq \cc(A)$ and $\monorect(M) \leq \monorect(A)$. 
\end{proof}

\subsection{Size optimal hazard-free formula}\label{subsec:sizeoptimalhazfree}
We now give the size optimal hazard-free formula for the multiplexer function. 
As a simple application of Theorem~\ref{thm:hf-kwthm}, 
we begin with finding optimal formulas for $\MUX_1$.

\begin{proposition}
\label{prop:mux1-optimal}
The optimal (size and depth) hazard-free De Morgan formula for $\MUX_1(s,x_0,x_1)$ has size $5$ and depth $3$. 
\end{proposition}
\begin{proof}
Consider the communication matrix of $\kw{\un}_{\MUX_1}$ shown below, 
\[
\begin{pNiceMatrix}[first-row,first-col]
~       & 00{\un}          &   {\un}00         & 1{\un}0 \\
01{\un} & \entry{x_0}      & \entry{x_0}       & \entry{s}   \\
{\un}11 & \entry{x_0}      & \entry{x_0, x_1}  & \entry{x_1} \\
1{\un}1 & \entry{s}        & \entry{x_1}       & \entry{x_1} 
\end{pNiceMatrix}
.\]
We find the following protocol for $\kw{\un}_{\MUX_1}$ by inspection:
\begin{equation}\label{eq:mux1monorect}
\begin{pNiceMatrix}[first-row,first-col,margin]
\CodeBefore
\cellcolor{blue!15}{2-3}
\cellcolor{yellow!30}{1-3}
\cellcolor{orange}{3-1}
\Body
~       & 00{\un}          &   {\un}00         & 1{\un}0 \\
01{\un} & \Block[fill=red!15]{2-2}{} x_0 & x_0 & s \\
{\un}11 & x_0 & x_0 & x_1   \\
1{\un}1 & s & \Block[fill=green!15]{1-2}{} x_1   & x_1 
\end{pNiceMatrix}
.\end{equation}
Using Lemma~\ref{lem:protocol-formula} with the above protocol, 
we obtain the hazard-free formula for $\MUX_1$ shown in Figure~\ref{fig:mux_formula}.  
The optimality of depth follows from the optimality of size. 
We defer the proof of the optimality of size to Theorem~\ref{thm:mux-lower}, the general case of $\MUX_n$. 
\end{proof}
\begin{remark}\label{rem:demystify}
To demystify the construction in Figure~\ref{fig:mux_formula} 
we note that it is simply the hazard-free DNF of $\MUX_1$, the formula \((s\wedge x_1) \vee (\neg s \wedge x_0) \vee (x_0 \wedge x_1)\),  
with an application of distributivity of $\wedge$ over $\vee$ 
to reduce the size. 
\end{remark}    

Now, using the recursive decomposition of $\MUX_n$ we obtain the 
following upper bound. 
\begin{theorem}
\label{thm:mux-upper}
The multiplexer function $\MUX_n$ has hazard-free
formulas 
of size $2 \cdot 3^n - 1$ and depth $3n$ for all $n \geq 1$.
\end{theorem}
\begin{proof}
We construct the formula inductively. The construction for $\MUX_1$ 
is given by Proposition~\ref{prop:mux1-optimal}. 
Recall that we can write $\MUX_n(s_1, \dotsc, s_{n}, x_0, \dotsc ,x_{2^n-1})$ recursively as the formula $F = \MUX_1(s_1, F_0, F_1)$, where
\[
F_0 =\MUX_{n-1}(s_2,\dotsc,s_{n},x_0,\dotsc,x_{2^{n-1}-1})
\ \text{ and } \
F_1 =\MUX_{n-1}(s_2,\dotsc,s_{n},x_{2^{n-1}},\dotsc,x_{2^n - 1}).
\]
By the induction hypothesis, both $F_0$ and $F_1$ have hazard-free formulas of size $2 \cdot 3^{n-1} - 1$ and depth $3(n-1)$. Using the hazard-free formula for $\MUX_1$, given in Figure~\ref{fig:mux_formula}, 
to implement $F$ yields a formula of size $2 \cdot 3^n - 1$ and depth $3n$ 
for $\MUX_n$. 

It remains to prove that the constructed formula $F$ is hazard-free. 
Using Lemma~\ref{lem:formula-protocol}, it suffices to show that 
the protocol using $F$ correctly solves the hazard-free $\kw{}$-game 
$\kw{\un}_{\MUX_n}$. In other words, the communication matrix of 
$\kw{\un}_{\MUX_n}$ is partitioned into monochromatic rectangles by the protocol given by $F$. 
We will prove it by induction on $n$. The base case, $n=1$, 
is given by Proposition~\ref{prop:mux1-optimal}. Now consider the inductive formula $F$ for $\MUX_n$ shown in Figure~\ref{fig:muxn_formula}. Following Lemma~\ref{lem:formula-protocol}, when Alice and Bob 
reach the 
colored nodes in $F$ (see Figure~\ref{fig:muxn_formula}), then 
we obtain the following partition of $M_{\kw{\un}_{\MUX_n}}$ as a block matrix, where $M_i := M_{\kw{\un}_{F_i}}$ for $i \in \{0, 1\}$, and $s_1$ stands for a block matrix of all entries $s_1$:
\[
\begin{pNiceMatrix}[first-row,first-col,margin]
\CodeBefore
\cellcolor{blue!15}{2-3}
\cellcolor{green!15}{3-2,3-3}
\cellcolor{yellow!30}{1-3}
\cellcolor{orange}{3-1}
\Body
~       & 0          &   {\un}         & 1 \\
0 & \Block[fill=red!15]{2-2}{} M_0 & M_0 & s_1 \\
{\un} & M_0 & M_0 & M_1   \\
1 & s_1 & M_1   & M_1 
\end{pNiceMatrix}
,\]
where the row (resp., column) labeled $\gamma \in \Tri$ represents the 
set of prime implicants (resp., implicates) with $s_1 = \gamma$. 
Now from the monochromatic partition of $M_i$ using $F_i$, $i\in\{0,1\}$, and Proposition~\ref{prop:parallel}, we get a  monochromatic partition of the communication matrix of
$\kw{\un}_{\MUX_n}$.
\begin{figure}
\centering
    \begin{tikzpicture}[xscale=0.6, yscale=0.8]
    \node [circle,draw,minimum size=1em,inner sep=1pt] {$\vee$} [sibling distance=3cm]
      child { node [circle,draw,minimum size=1em,inner sep=1pt] {$\wedge$} [sibling distance=2cm]
        child {node [circle,draw,fill=red!15,minimum size=1em,inner sep=1pt] {$F_0$}}
        child {node [circle,draw,minimum size=1em,inner sep=1pt] {$\vee$}
          child {node [circle,draw,fill=blue!15,minimum size=1em,inner sep=1pt] {$F_1$}}
          child {node [circle,draw,fill=yellow!30,minimum size=1em,inner sep=1pt] {$\neg s_1$}}}}
      child {node [circle,draw,minimum size=1em,inner sep=1pt] {$\wedge$} [sibling distance=2cm]
        child {node [circle,draw,fill=green!15,minimum size=1em,inner sep=1pt] {$F_1$}}
        child {node [circle,draw,fill=orange,minimum size=1em,inner sep=3pt] {$s_1$}}};
    \end{tikzpicture}
    \caption{$F$: a size-optimal hazard-free formula for $\MUX_n$}
    \label{fig:muxn_formula}
\end{figure}
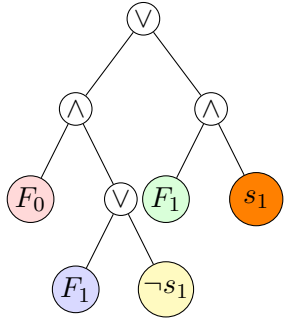
\end{proof}

\begin{remark}
We remark that the construction in Theorem~\ref{thm:mux-upper} can be made into an \emph{alternating} formula:
the communication matrix is symmetric and hence the monochromatic partition can be transposed so that Bob starts the communication.
We now use this transposed partition on the blue $F_1$ node in Figure~\ref{fig:muxn_formula}.
\end{remark}

We now prove that the above construction for $\MUX_n$ is optimal 
with respect to size. For this purpose, we study the communication problem 
associated with the following \emph{subcube intersection} function, 
\begin{align*}
    \textsf{subcube-intersect}_n\colon \Tri^n \times \Tri^n \to \Bool,
\end{align*}
where $\textsf{subcube-intersect}_n(\alpha,\beta) = 1$ if and only if
the subcubes defined by $\alpha$ and $\beta$ in $\Bool^n$ intersect, i.e., if $\alpha$ and $\beta$ have a common resolution. We note that the subcube intersection function is the same as the equality function when
restricting its domain of definition to Boolean values only.
The equality function is widely used in classical communication complexity for proving lower bounds. We also note that the subcube intersection function
cannot be implemented by any circuit over $\{0,\un,1\}$
(and hence in particular is not the hazard-free extension of any Boolean function), even for $n=1$, because
$\textsf{subcube-intersect}_1(\un,\un) = 1$, but
$\textsf{subcube-intersect}_1(0,1) = 0$ \footnote{In any circuit implementation, if $C(\alpha)=1$, then for all resolutions $a$ of $\alpha$ we also have $C(a)=1$, which is easily seen by induction. Alternatively, this can be seen by the fact that all gates (and hence the whole circuit) are monotone with respect to the partial order of stability ($\un \sqsubseteq 0$, $\un \sqsubseteq 1$, $0$ and $1$ incomparable), so switching unstable inputs to stable inputs can only keep an output $\un$ or switch an output from $\un$ to a stable value, but not change a stable output.}.
Let us see how the subcube intersection problem 
helps in capturing the complexity of the hazard-free game $\kw{\un}_{\MUX_n}$. 
\begin{lemma}
\label{lem:subcube-intersect-mux}
The \textup{$\textsf{subcube-intersect}_n$} communication problem reduces to the communication problem $\kw{\un}_{\MUX_n}$ with no extra cost. 
\end{lemma}
\begin{proof}
Given inputs $\alpha,\beta \in \Tri^n$ to the \textsf{subcube-intersect}$_n$ 
problem, Alice and Bob modify their inputs as follows \emph{without} communication. 
\begin{compactitem}
    \item Alice constructs $\alpha' \in \{\un,1\}^{2^n}$ such that $\alpha'$ 
    has ones only at the positions indexed by the subcube of resolutions of $\alpha$. 
    \item Bob constructs $\beta'\in\{\un,0\}^{2^n}$ such that $\beta'$ 
    has zeroes only at the positions indexed by the subcube of resolutions of $\beta$. 
\end{compactitem}
Now they can solve the game $\kw{\un}_{\MUX_n}$ on inputs $\alpha\alpha'$ and $\beta\beta'$. Observe that if the subcubes $\alpha$ and $\beta$ intersect 
then answers to $\kw{\un}_{\MUX_n}$ lie in the set of data variables $\{x_0,\dotsc ,x_{2^n-1}\}$, otherwise they lie in the set of selector 
variables $\{s_1, \dotsc , s_n\}$. Therefore, from the answers to the $\kw{\un}_{\MUX_n}$ game they can deduce whether the subcubes intersect or not, 
again without communication. 
\end{proof}

Using the rank lower bound technique of \cite{MS:82} (See also \cite[Lemma~1.28]{KN96} and the discussion following the lemma.), we know that 
\begin{equation}\label{eq:MehohornSchmidt}
\psize(\textsf{subcube-intersect}_n)\geq2\cdot\rank(M_{\textsf{subcube-intersect}_n})-1,
\end{equation}
where $M_{\textsf{subcube-intersect}_n}$ is interpreted as a matrix over $\mathbb R$ with 0s and 1s as entries.
We prove the following tight bound on the rank of $M_{\textsf{subcube-intersect}_n}$:
\begin{lemma}
\label{lem:subcube-intersect}
The communication matrix of \textup{$\textsf{subcube-intersect}_n$} is of \emph{full} rank. 
That is, the rank of $M_{\textup{\textsf{subcube-intersect}$_n$}}$
equals $3^n$ for all $n\geq1$.
\end{lemma}

This immediately implies our size lower bound:
\begin{theorem}
\label{thm:mux-lower}
  Any hazard-free
formula for $\MUX_n$ requires $2 \cdot 3^n - 1$ leaves for all $n\geq 1$.
\end{theorem}
\begin{proof}
Using Theorem~\ref{thm:hf-kwthm}, it is sufficient to show that the
communication matrix of $\kw{\un}_{\MUX_n}$ requires 
$2\cdot3^n-1$ monochromatic rectangles, i.e., 
\( \psize(\kw{\un}_{\MUX_n}) \geq 2\cdot 3^n-1.\)
This is readily checked:
\begin{align*}
\monorect(\kw{\un}_{\MUX_n}) &\stackrel{\text{Lem.~\ref{lem:subcube-intersect-mux}}}{\geq} \monorect(\textsf{subcube-intersect}_n)
\\
&\stackrel{\eqref{eq:MehohornSchmidt}}{\geq} 2\cdot\rank(M_{\textsf{subcube-intersect}_n})-1 \stackrel{\text{Lem.~\ref{lem:subcube-intersect}}}{=} 2 \cdot 3^n -1.\qedhere
\end{align*}

\end{proof}

It now remains to prove Lemma~\ref{lem:subcube-intersect}. 
\begin{proof}[Proof of Lemma~\ref{lem:subcube-intersect}] 
We prove it by induction on $n$. For the base case, $n=1$ and the 
communication matrix $M_{\textsf{subcube-intersect}_1}$ is as follows:
\[
\begin{pNiceMatrix}[first-row,first-col,margin]
~       & 0          &   {\un}         & 1 \\
0 &  1 & 1 & 0 \\
{\un} & 1 & 1 & 1   \\
1 & 0 & 1   & 1 
\end{pNiceMatrix}
.\]
Clearly $\rank(M_{\textsf{subcube-intersect}_1}) = 3$. Now consider the communication matrix of $M_{\textsf{subcube-intersect}_n}$. We claim that it looks as follows:
\[
\begin{pNiceMatrix}[first-row,first-col,margin]
~       & 0          &   {\un}         & 1 \\
0 &  M_{\textsf{subcube-intersect}_{n-1}} & M_{\textsf{subcube-intersect}_{n-1}} & 0 \\
{\un} & M_{\textsf{subcube-intersect}_{n-1}} & M_{\textsf{subcube-intersect}_{n-1}} & M_{\textsf{subcube-intersect}_{n-1}}   \\
1 & 0 & M_{\textsf{subcube-intersect}_{n-1}}   & M_{\textsf{subcube-intersect}_{n-1}} 
\end{pNiceMatrix}
,\]
where the row labeled $\gamma \in \Tri$ represents the 
set of rows labeled with $\alpha \in \Tri^n$ such that
$\alpha_1 = \gamma$ and similarly for the columns.
The validity of the claim follows from inspection that on fixing 
the first variables we either know the answer or 
have self-reduced it to a smaller instance. 
Therefore, we obtain:
\[M_{\textsf{subcube-intersect}_n} = 
\begin{pNiceMatrix}[margin]
  1 & 1 & 0 \\
 1 & 1 & 1   \\
 0 & 1   & 1 
\end{pNiceMatrix}
\otimes M_{\textsf{subcube-intersect}_{n-1}}, \]
where $\otimes$ is the Kronecker product of matrices.  
Hence, using the fact that rank is multiplicative with respect to Kronecker product, we have 
\[\rank(M_{\textsf{subcube-intersect}_n}) = \rank(M_{\textsf{subcube-intersect}_1})\cdot\rank(M_{\textsf{subcube-intersect}_{n-1}}).\]
Now using the induction hypothesis completes the proof.  
\end{proof}

Translating the size lower bound to depth gives the following corollary. 
\begin{corollary}
\label{cor:hf-depth-lb}
For all $n\geq 1$, $\depth{\un}(\MUX_n) \geq \lceil \log_2(2 \cdot 3^n - 1) \rceil$.  
\end{corollary}

\subsection{Formulas of improved depth}
The lower bound from Corollary~\ref{cor:hf-depth-lb} on hazard-free formula depth is at least $1+(\log_2 3)\cdot n$ for large $n$. 
However, our construction in Theorem~\ref{thm:mux-upper} gives an upper bound of $3n$. 
We now give an improved construction (Theorem~\ref{thm:mux-depth-upper}) with respect to depth 
while increasing the size by a factor of
$\frac{9}{8}$.
In contrast, the depth-optimal version of Huffman's construction (Proposition~\ref{prop:depthlowerbound}) is larger than the optimal size hazard-free formula by a multiplicative factor that is exponential in $n$.

\begin{theorem}
\label{thm:mux-depth-upper}
The multiplexer function $\MUX_n$ has hazard-free
formulas 
of depth $2n+1$ and size at most
$2.25\cdot 3^n -\frac{n}{2} - 1.25$ for all $n \geq 1$.
\end{theorem}
From Theorem~\ref{thm:hf-kwthm} we know it is sufficient to 
give a protocol $\Pi$ solving the hazard-free Karchmer-Wigderson game of 
$\MUX_n$ such that $\cc(\Pi)\leq2n+1$ and $\psize(\Pi)\leq 2.25\cdot 3^n -\frac{n}{2} - 1.25$. 
We consider a monochromatic extension of 
$\kw{\un}_{\MUX_n}$.
We extend the communication matrix of 
$\kw{\un}_{\MUX_n}$ as follows to define the 
extended version $\textsf{e-}\kw{\un}_{\MUX_n}$: 
\begin{align}
\label{eq:def-e-kw-mux}
\begin{pNiceMatrix}[hlines,first-row,first-col,margin]
~       &           &   \text{prime implicates}         &  \\
 & \Block{3-3}{M_{\kw{\un}_{\MUX_n}}}  &  &  \\
\text{prime implicants} &  &  &    \\
 &  &   &  \\
& \Block[]{2-3}<\Large>{0} & & \\
& & & 
\end{pNiceMatrix}
.
\end{align}
In other words, the communication matrix of the extended version 
$\textsf{e-}\kw{\un}_{\MUX_n}$ is obtained by adding 
a rectangular block 
with all $0$s to either the set of rows of $\kw{\un}_{\MUX_n}$ 
(as shown above) or the set of columns. 
We note that the added block could have any number of rows 
in the former case or any number of columns in the latter.  
Further we will \emph{always} require
that the extended part be filled with a number that does not appear in $\kw{\un}_{\MUX_n}$.  
Observe that $0$ doesn't appear in the communication matrix of $\kw{\un}_{\MUX_n}$. However we could have used any other number that 
doesn't appear in $\kw{\un}_{\MUX_n}$. We will denote all such extensions 
by $\textsf{e-}\kw{\un}_{\MUX_n}$.
The following matrix is an example of $\textsf{e-}\kw{\un}_{\MUX_2}$. 
\[
\begin{pNiceMatrix}[margin,first-row,first-col]
{} & 000\un\un\un & 0\un00\un\un & 01{\un}0{\un}{\un} & \un 00{\un}0{\un} & \un\un 0000 & \un 1{\un}0{\un}0 & 10{\un}{\un}0{\un} & 1\un{\un}{\un}00 & 11{\un}{\un}{\un}0 \\
001{\un}{\un}{\un} & x_{0} & x_{0} & s_2 & x_{0} & x_{0} & s_2 & s_1 & s_1 & s_1, s_2 \\
0{\un}11\un\un & x_{0} & x_{0}, x_{1} & x_{1} & x_{0} & x_{0},x_{1} & x_{1} & s_1 & s_1 & s_1 \\
01\un1\un\un & s_2 & x_{1} & x_{1} & s_2 &x_{1}& x_{1} & s_1, s_2 & s_1 & s_1 \\
\un 01\un1\un& x_{0} & x_{0} & s_2 & x_{0}, x_{2} &x_{0},x_{2}&  s_2 & x_{2} & x_{2} & s_2 \\
\un\un 1111&  x_{0} & x_{0},x_{1} & x_{1} & x_{0},x_{2}&x_{0},x_{1},x_{2},x_{3}& x_{1},x_{3} & x_{2} & x_{2},x_{3}& x_{3}\\
\un 1\un1\un1 & s_2 & x_{1} & x_{1} & s_2 &x_{1},x_{3}& x_{1}, x_{3} & s_2 & x_{3} & x_{3} \\
10\un\un1\un & s_1 & s_1 & s_1, s_2 & x_{2} &x_{2}& s_2 & x_{2} & x_{2} & s_2 \\
1\un\un\un11 & s_1 & s_1 & s_1 & x_{2} &x_{2},x_{3}& x_{3} & x_{2} & x_{2}, x_{3} & x_{3} \\
11\un\un\un1 & s_1, s_2 & s_1 & s_1 & s_2 &x_{3}& x_{3} & s_2 & x_{3} & x_{3} \\
  & 0 & 0 & 0 & 0 & 0 & 0 & 0 & 0 & 0
\end{pNiceMatrix}
\]

Clearly the following proposition holds. 
\begin{proposition}
\label{prop:e-kw-mux}
A protocol $\Pi$ for $\mathsf{e\text{-}}\kw{\un}_{\MUX_n}$ 
gives a protocol $\Pi'$ for $\kw{\un}_{\MUX_n}$ 
such that $\cc(\Pi') \leq \cc(\Pi)$  and 
$\monorect(\Pi') \leq \monorect(\Pi)$. 
\end{proposition}
\begin{proof}
Follows from the definition \eqref{eq:def-e-kw-mux} of $\textsf{e-}\kw{\un}_{\MUX_n}$. 
\end{proof}
Therefore, to 
prove the depth bound in Theorem~\ref{thm:mux-depth-upper}
we will give a protocol 
for $\textsf{e-}\kw{\un}_{\MUX_n}$ with communication cost at most $2n+1$. 
\begin{lemma}
\label{lem:e-kw-mux-ub}
There is a protocol solving $\textup{\textsf{e-}}\kw{\un}_{\MUX_n}$ such that its 
communication cost is at most $2n+1$. 
\end{lemma}
\begin{proof}
From Proposition~\ref{prop:mux-comm-matrix} we know that 
the communication matrix of $\kw{\un}_{\MUX_n}$ looks as follows 
\[
\begin{pNiceMatrix}[first-row,first-col,margin]
~       & 0          &   {\un}         & 1 \\
0 & M_0 & M_0 & s_1 \\
{\un} & M_0 & M_0\cup M_1 & M_1   \\
1 & s_1 & M_1   & M_1 
\end{pNiceMatrix}
,\]
where we define the formulas $F_0 =\MUX_{n-1}(s_2,\dotsc,s_{n},x_0,\dotsc,x_{2^{n-1}-1})$ 
and $F_1 =\MUX_{n-1}(s_2,\dotsc,s_{n},x_{2^{n-1}},\dotsc,x_{2^n - 1})$ and matrices $M_i := M_{\kw{\un}_{F_i}}$ for $i \in \{0, 1\}$.
Therefore, the matrix of extended version $\textsf{e-}\kw{\un}_{\MUX_n}$ looks as follows 
\[
\begin{pNiceMatrix}[first-row,first-col,margin]
~       & 0          &   {\un}         & 1 \\
0 & M_0 & M_0 & s_1 \\
{\un} & M_0 & M_0\cup M_1 & M_1   \\
1 & s_1 & M_1   & M_1  \\
& 0 & 0 & 0
\end{pNiceMatrix}
.\]
We now give a protocol to partition this matrix into monochromatic rectangles. This will be done inductively. 

Base case: $n=1$. The matrix of $\textsf{e-}\kw{\un}_{\MUX_1}$ 
can be monochromatically partitioned as follows
\[
\begin{tikzpicture}
\node at (0,0) {$
\begin{pNiceArray}{ccc}[first-row,first-col,rules/width=3pt]
~       & 00{\un}          &   {\un}00         & 1{\un}0 \\
\arrayrulecolor{blue}
01{\un} & \entry{x_0}      & \entry{x_0}       & \multicolumn{1}{|c}{\entry{s}}   \\ 
\arrayrulecolor{orange}\cline{3-3}\arrayrulecolor{blue}
{\un}11 & \entry{x_0}      & \entry{x_0, x_1}  & \multicolumn{1}{|c}{\entry{x_1}} \\ 
\arrayrulecolor{red}\hline\arrayrulecolor{blue} 
1{\un}1 & \entry{s}        & \multicolumn{1}{|c}{\entry{x_1}}       & \entry{x_1} \\
\arrayrulecolor{orange}\cline{1-1}\cline{2-3}\arrayrulecolor{blue}
& 0 & \multicolumn{1}{|c}{0} & 0
\end{pNiceArray}
$};
\fill [blue] (-0.35,-0.7) rectangle +(0.1,-0.2);
\fill [blue] (0.94,0.45) rectangle +(0.1,-0.2);
\node at (2.2,-0.3) {.};
\end{tikzpicture}
\]
This is obtained from the protocol where 
Alice sends the first bit indicating whether her input lies in the 
top or bottom part of red line. Then Bob sends a bit indicating whether 
his input lies in the left part or the right part of the blue line(s). 
Finally Alice sends one last bit to indicate whether her input 
lies in the top or bottom part of the orange line(s). 
Clearly the communication cost is $3$. 

Induction step: $n\geq 2$. 
Now Alice and Bob send one bit of communication each to reduce 
$\textsf{e-}\kw{\un}_{\MUX_n}$ to the following partition
\[
\begin{pNiceArray}{ccc}[first-row,first-col,rules/width=3pt]
~       & 0          &   {\un}         & 1 \\
\arrayrulecolor{blue}
0 & M_0 & M_0 & \multicolumn{1}{|c}{s_1} \\
{\un} & M_0 & M_0\cup M_1 & \multicolumn{1}{|c}{M_1}   \\
\arrayrulecolor{red}\hline\arrayrulecolor{blue}
1 & s_1 & \multicolumn{1}{|c}{M_1}   & M_1  \\
& 0 & \multicolumn{1}{|c}{0} & 0
\end{pNiceArray}
.\]
Observe that the top-right block $\begin{pmatrix}s_1\\M_1\end{pmatrix}$ 
is the matrix of an $\textsf{e-}\kw{\un}_{\MUX_{n-1}}$. 
The bottom-right block $\begin{pmatrix}M_1 & M_1 \\ 0 & 0\end{pmatrix}$ 
can be solved, using Proposition~\ref{prop:parallel}, at the cost of solving $\begin{pmatrix}M_1\\0\end{pmatrix}$ which is a block of $\textsf{e-}\kw{\un}_{\MUX_{n-1}}$. Similarly, the top-left block
$\begin{pmatrix}M_0&M_0\\M_0&M_0\cup M_1\end{pmatrix}$ can be solved, using Proposition~\ref{prop:parallel}, at the cost of solving $M_0$ which is a block of $\kw{\un}_{\MUX_{n-1}}$. Therefore, by Proposition~\ref{prop:e-kw-mux} it has less complexity (size and depth) 
than $\textsf{e-}\kw{\un}_{\MUX_{n-1}}$. 
Finally, note that the bottom-left block $\begin{pmatrix}s_1\\0\end{pmatrix}$ 
just needs one bit of communication to monochromatically partition it. 
Therefore, we have 
\begin{align*}
    \cc(\mathsf{e\text{-}}\kw{\un}_{\MUX_{n}}) & \leq 
    2 + \max\left\{ \cc(\mathsf{e\text{-}}\kw{\un}_{\MUX_{n-1}}) , 1\right\} \\ 
    & \leq 2+ \cc(\mathsf{e\text{-}}\kw{\un}_{\MUX_{n-1}}), \\
    & \leq 2n + 1,
\end{align*}
where the second inequality follows because $n\geq 2$ and the third 
follows from the induction hypothesis. 
\end{proof}
As an illustration of the induction step in the proof above, we present a detailed example of $\textsf{e-}\kw{\un}_{\MUX_2}$ with a decomposition
into one block of $\textsf{e-}\kw{\un}_{\MUX_1}$ (red),
two blocks of $\textsf{e-}\kw{\un}_{\MUX_1}$ (green) that can be solved at the cost of solving a single $\textsf{e-}\kw{\un}_{\MUX_1}$ using Proposition~\ref{prop:parallel},
four blocks of $\MUX_1$ (blue) that can be solved at the cost of solving a single $\MUX_1$ using Proposition~\ref{prop:parallel},
and one (yellow) block of two identities that can be solved with depth 1. 
Therefore, the total depth is at most 5. 
\[
\begin{pNiceMatrix}[margin,first-row,first-col]
\CodeBefore
\tikz \draw [fill=red!50] (1-|7) |- (7-|7) |- (7-|10) |- cycle ;
\tikz \draw [fill=blue!15] (1-|1) |- (4-|1) |- (4-|4) |- cycle ;
\tikz \draw [fill=blue!15] (1-|4) |- (4-|4) |- (4-|7) |- cycle ;
\tikz \draw [fill=blue!15] (4-|1) |- (7-|1) |- (7-|4) |- cycle ;
\tikz \draw [fill=blue!15] (4-|4) |- (7-|4) |- (7-|7) |- cycle ;
\tikz \draw [fill=green!50] (7-|7) |- (11-|7) |- (11-|10) |- cycle ;
\tikz \draw [fill=green!50] (7-|4) |- (11-|4) |- (11-|7) |- cycle ;
\tikz \draw [fill=yellow!50] (7-|1) |- (11-|1) |- (11-|4) |- cycle ;
\Body
{} & 000\un\un\un & 0\un00\un\un & 01{\un}0{\un}{\un} & \un 00{\un}0{\un} & \un\un 0000 & \un 1{\un}0{\un}0 & 10{\un}{\un}0{\un} & 1\un{\un}{\un}00 & 11{\un}{\un}{\un}0 \\
001{\un}{\un}{\un} & x_{0} & x_{0} & s_2 & x_{0} & x_{0} & s_2 & s_1 & s_1 & s_1, s_2 \\
0{\un}11\un\un & x_{0} & x_{0}, x_{1} & x_{1} & x_{0} & x_{0},x_{1} & x_{1} & s_1 & s_1 & s_1 \\
01\un1\un\un & s_2 & x_{1} & x_{1} & s_2 &x_{1}& x_{1} & s_1, s_2 & s_1 & s_1 \\
\un 01\un1\un& x_{0} & x_{0} & s_2 & x_{0}, x_{2} &x_{0},x_{2}&  s_2 & x_{2} & x_{2} & s_2 \\
\un\un 1111&  x_{0} & x_{0},x_{1} & x_{1} & x_{0},x_{2}&x_{0},x_{1},x_{2},x_{3}& x_{1},x_{3} & x_{2} & x_{2},x_{3}& x_{3}\\
\un 1\un1\un1 & s_2 & x_{1} & x_{1} & s_2 &x_{1},x_{3}& x_{1}, x_{3} & s_2 & x_{3} & x_{3} \\
10\un\un1\un & s_1 & s_1 & s_1, s_2 & x_{2} &x_{2}& s_2 & x_{2} & x_{2} & s_2 \\
1\un\un\un11 & s_1 & s_1 & s_1 & x_{2} &x_{2},x_{3}& x_{3} & x_{2} & x_{2}, x_{3} & x_{3} \\
11\un\un\un1 & s_1, s_2 & s_1 & s_1 & s_2 &x_{3}& x_{3} & s_2 & x_{3} & x_{3} \\
  & 0 & 0 & 0 & 0 & 0 & 0 & 0 & 0 & 0
\end{pNiceMatrix}
\]

We are now ready to prove Theorem~\ref{thm:mux-depth-upper}. 
\begin{proof}[Proof of Theorem~\ref{thm:mux-depth-upper}]
As mentioned in the beginning we will give a protocol to partition 
the communication matrix of $\kw{\un}_{\MUX_n}$ into monochromatic rectangles 
such that the communication cost of this protocol is at most $2n+1$ and 
the number of monochromatic rectangles is at most 
$2.25\cdot 3^n -\frac{n}{2} - 1.25$ for all $n \geq 1$.
Our protocol is the same as the one given in the proof of Lemma~\ref{lem:e-kw-mux-ub}. Thus the upper bound on depth follows readily. To bound the size we count the number of monochromatic rectangles in the partition given by the protocol. 

Define $T(n)$ to be the number of monochromatic rectangles 
in the partition of the communication matrix of $\textsf{e-}\kw{\un}_{\MUX_n}$ given by the protocol. 
Further define $S(n)$ to be the number of monochromatic rectangles covering only the entries of $\kw{\un}_{\MUX_n}$ in the partition of 
the communication matrix of $\textsf{e-}\kw{\un}_{\MUX_n}$. 
In particular, $S(n) < T(n)$. Note that $S(n)$ gives the bound on the size we are interested in. 
We now write recurrences for $S(n)$ and $T(n)$ using the partition given by the induction step and the base case in Lemma~\ref{lem:e-kw-mux-ub}. 
Recall the partition in the induction step looks as shown below: 
\[
\begin{pNiceArray}{ccc}[first-row,first-col,rules/width=3pt]
~       & 0          &   {\un}         & 1 \\
\arrayrulecolor{blue}
0 & M_0 & M_0 & \multicolumn{1}{|c}{s_1} \\
{\un} & M_0 & M_0\cup M_1 & \multicolumn{1}{|c}{M_1}   \\
\arrayrulecolor{red}\hline\arrayrulecolor{blue}
1 & s_1 & \multicolumn{1}{|c}{M_1}   & M_1  \\
& 0 & \multicolumn{1}{|c}{0} & 0
\end{pNiceArray}
.\]
It follows that $S(n)$ and $T(n)$ satisfies the following recurrences for $n\geq 2$:
\begin{align}
S(n) & = 2\cdot S(n-1) + T(n-1) + 1, \mbox{ and }\label{eq:recurrence-S(n)}\\
T(n) &= S(n-1) + 2\cdot T(n-1) + 2,\label{eq:recurrence-T(n)}
\end{align}
where $S(1) = 5$ and $T(1) = 7$. 
For solving these recurrences it is helpful to first solve the two auxiliary recurrences $S(n)+T(n)$ and $S(n)-T(n)$. One ultimately obtains $S(n) = 2.25\cdot 3^n -\frac{n}{2} - 1.25$.
\end{proof}
\begin{remark}
We note that in the induction step we are reducing to an instance of $\textsf{e-}\kw{\un}_{\MUX}$, in particular the top-right part of the matrix in the proof of Theorem~\ref{thm:mux-depth-upper}.  
Therefore, we need to strengthen our induction hypothesis and begin with an instance of $\textsf{e-}\kw{\un}_{\MUX_{n}}$. 
\end{remark}

\section{Alternation Depth Two and Two-round Protocols}
\label{sec:tworoundprot}
In this section, we determine the hazard-free depth complexity of formulas of alternation depth $2$ computing $\MUX_n$.
We assume without loss of generality that the output gate is an $\vee$ gate since the prime implicants and prime implicates of $\MUX_n$ are symmetric.
We exploit the correspondence between hazard-free formulas of alternation depth $2$ and two-round communication protocols of the form: Alice sends some string, Bob replies with some string, and they settle on an answer (See Lemmas~\ref{lem:formula-protocol} and \ref{lem:protocol-formula}).

We begin by proving a property of hazard-free formulas of alternation depth $2$ that has been observed in \cite{Huf:57}. We present a proof for completeness.

\begin{proposition}\label{prop:dnf-unique}
Let $n \geq 1$. In any two-round communication protocol for $\MUX_n$, for different prime implicants, Alice must send different strings to Bob in the first round. Equivalently, in any hazard-free formula of alternation depth $2$ computing $\MUX_n$, for any prime implicant $\alpha$, there is at least one $\wedge$ gate $g$ in the formula such that the subformula at $g$ is exactly the $\wedge$ of literals in the prime implicant~$\alpha$.
\end{proposition}
\begin{proof}
Let $\alpha, \beta \in \Tri^n$ be two distinct prime implicants of $\MUX_n$ (See Proposition~\ref{prop:mux-prime}). Then Bob must receive different strings from Alice for these two inputs to Alice, which can be seen as follows. Suppose Bob receives the same string from Alice for prime implicants $\alpha$ and $\beta$. Without loss of generality (i.e., we can swap $\alpha$ and $\beta$) $\alpha \neq \un^n$ and there exists a position $i$ such that $\{0,\un,1\}\ni \beta_i \neq \alpha_i\in\{0, 1\}$. Let $\alpha'$ be the prime implicate obtained from $\alpha$ by flipping the $i^\text{th}$ bit.
Clearly, the only answer to the input $(\alpha, \alpha')$ is $i$ and $i$ is a wrong answer for the input $(\beta, \alpha')$. Therefore, Alice must send different strings to Bob for $\alpha$ and $\beta$.

Now, we prove the equivalent statement for formulas. Consider an arbitrary hazard-free formula~$F$ for $\MUX_n$ of alternation depth $2$. For any prime implicant $\alpha$ of $\MUX_n$, since $F$ is hazard-free, we have $F(\alpha) = 1$. This is possible only if there is a topmost $\wedge$ gate $g$ in $F$ such that $g(\alpha) = 1$. Since the formula has alternation depth $2$, the subformula at $g$ is simply an $\wedge$ of literals. Since it evaluates to $1$ on $\alpha$, this set of literals has to be a subset of the literals in $\alpha$. It cannot be a proper subset because of minimality of prime implicants.
\end{proof}

We now observe a slightly weaker (than Theorem~\ref{thm:mux-dnf-depth}) depth lower bound using known results. We define the \emph{size} of a prime implicant $\alpha$ as $|\{i\mid \alpha_i \neq \un\}|$.

\begin{proposition}\label{prop:depthlowerbound}
For $n \geq 1$, we have 
$\size{\un}_{2}(\MUX_n) = 4^n + 2n3^{n-1}$ and 
$\depth{\un}_{2}(\MUX_n) \geq 2n+1.$
\end{proposition}
\begin{proof}
Consider an arbitrary hazard-free formula $F$ of alternation depth $2$ computing $\MUX_n$. By Proposition~\ref{prop:dnf-unique}, for each prime implicant $\alpha$ of $\MUX_n$, there is at least one subformula that computes $\alpha$. Therefore, the size of $F$ must be at least the sum of the sizes of all prime implicants. We now compute the sum of the sizes of all prime implicants of $\MUX_n$. The size of any prime implicant where the selector bits have exactly $i$ $\un$'s is $n - i + 2^i$. This is because the other $n-i$ selector bits must have Boolean values, and the subcube indexed by these selector bits contains $2^i$ points, which must all have value 1 in the data bits. There are $\binom{n}{i}$ ways to choose $i$ positions for the selector bits with unstable values. For each such choice, there are $2^{n-i}$ ways to set the remaining selector bits. Therefore, the sum of the sizes of all prime implicants in $\MUX_n$ is $\sum_{i = 0}^{n} \binom{n}{i} 2^{n-i} (2^i + n-i) = 4^n + 2n3^{n-1}$, giving us the required size lower bound (This is also an upper bound using \cite{Huf:57}). We readily conclude $\depth{\un}_{2}(f) \geq \lceil \log_2(4^n + 2n3^{n-1}) \rceil = 2n+1$.
\end{proof}

We now show that the existence of low-depth hazard-free formulas of alternation depth $2$ is linked to the existence of certain short prefix codes.

\begin{lemma}\label{lem:hf-prefix-code}
Let $n \geq 1$. For $d \geq 0$, there is a hazard-free formula of alternation depth $2$ and depth $d$ for $\MUX_n$ if and only if there is a prefix code for the set of all prime implicants of $\MUX_n$ such that for any $i$ where $0 \leq i \leq n$, each prime implicant with exactly $i$ $\un$'s in the selector bits is encoded using at most $d - \lceil \log_2(2^i + n - i) \rceil$ bits.
\end{lemma}
\begin{proof}
Let $i$ be the number of $\un$'s in the selector bits of the prime implicant given to Alice. First, we determine a tight bound on the number of bits that Bob must transfer based on $i$ (after Alice has transferred her prime implicant to Bob). Once Bob has received the prime implicant from Alice, the final answer could be any of the selector bits with Boolean values ($n-i$ possibilities) or any of the data bits in the subcube indexed by Alice's selector bits ($2^i$ possibilities). Since there are $2^i + n - i$ distinct answers possible, Bob must use at least $\lceil \log_2(2^i + n - i) \rceil$ bits in his reply. This bound is tight. Once Bob receives the prime implicant from Alice, he can reply with the answer using at most $\lceil \log_2(2^i + n - i) \rceil$ bits. Therefore, a two-round protocol of depth $d$ exists if and only if there is a prefix code for the prime implicants that uses at most $d - \lceil \log_2(2^i + n - i) \rceil$ bits to encode prime implicants with $i$ $\un$'s in the selector bits.
\end{proof}

We now prove that the optimal depth for hazard-free formulas with alternation depth $2$ is $2n+2$.
This is the only depth lower bound in this paper that does not follow directly from a size lower bound.
We will need the well-known Kraft's inequality giving a necessary and sufficient condition for the existence of a prefix code. 

\begin{theorem}[{\cite[Theorem 5.2.1]{CT91}}]
\label{thm:kraft-ineq}
For any binary prefix code, the codeword lengths $\ell_1,\dotsc,\ell_m$ must satisfy the inequality 
\[\sum_{i=1}^m 2^{-\ell_i} \leq 1.\]
Conversely, given a set of codeword lengths that staisfy this inequality, there exists a prefix code with these codeword lengths. 
\end{theorem}
\begin{theorem}\label{thm:mux-dnf-depth}
For $n \geq 2$, we have $\depth{\un}_{2}(\MUX_n) = 2n+2$.
\end{theorem}
\begin{proof}
From Lemma~\ref{lem:hf-prefix-code} we know that it suffices to find the minimal $d$ for which there exists a prefix code for the set of all prime implicants of $\MUX_n$ such that for any $i$, $0\leq i\leq n$, each prime implicant with exactly $i$ $\un$'s in the selector bits is encoded using at most $d-\lceil\log_2(2^{i}+n-i)\rceil$ bits. We know that there are $\binom{n}{i}2^{n-i}$ many prime implicants with exactly $i$ $\un$'s. Now using Kraft's inequality (Theorem~\ref{thm:kraft-ineq}) we have that the lengths of the encoding for each prime implicant must satisfy the following inequality 
\[\sum_{i=0}^n\binom{n}{i}2^{n-i}\cdot 2^{-(d-\lceil\log_2(2^{i}+n-i)\rceil)}  \leq 1.\]
Rearranging we obtain
\[\sum_{i=0}^n\binom{n}{i}2^{n-i}\cdot 2^{\lceil\log_2(2^{i}+n-i)\rceil} \leq 2^d.\]

For a fixed $n$, define the function $\Psi(i) := \lceil\log_2(2^{i}+n-i)\rceil - i$, where $0\leq i\leq n$. 
Note that for all $i \in \{0,1,\dotsc, n\}$, 
$\Psi(i)\geq 0$ and further $\Psi(i)$ is an integer. 
Now consider $\psi(i) := \log_2(2^{i}+n-i) - i$. 
Clearly from the definitions we have $\Psi(i) = \lceil\psi(i)\rceil$ for all $i \in \{0,1,\dotsc, n\}$.  From elementary calculus it follows 
that $\psi(i)$ is a continuous function that is decreasing in the interval $[0,n]$. Furthermore observe that at 
$i=0$, $\psi(0) = \log_2(n+1)$ and $\Psi(0) = \lceil\log_2(n+1)\rceil$, and at $i=n-1$, 
$\Psi(n-1)=\lceil\psi(n-1)\rceil=1$. 
Therefore
when $n\geq 2$, $\Psi(0) \geq 2$, and thus by the continuity of $\psi$ there exists a $t \in \{1,\dotsc n-1\}$ such that for all $i\in \{t,\dotsc,n-1\}$, 
$\Psi(i) = 1$. Choose the minimal such $t$. In the following we will work with this minimal $t$.  

Now breaking the summation at $t-1$ we get 
\[ \sum_{i=0}^{t-1}\binom{n}{i}2^{n-i}\cdot 2^{\Psi(i)+i} + \sum_{i=t}^{n-1}\binom{n}{i}2^{n-i}\cdot 2^{\Psi(i)+i} + \binom{n}{n}2^n \leq 2^d.\]
Dividing by $2^{n+1}$ on both sides we obtain 
\begin{align}
\label{eq:kraft-bound}
\sum_{i=0}^{t-1}\binom{n}{i} 2^{\Psi(i)-1} + \sum_{i=t}^{n-1}\binom{n}{i} 2^{\Psi(i)-1} + \binom{n}{n}\frac{1}{2} \leq 2^{d-n-1}.
\end{align}

Up to this point we only reformulated the property of the existence of a prefix code with the desired properties. We now prove the theorem by first considering the lower bound and then the upper bound.

Now by the minimality of $t$ and the integrality of $\Psi$, for $0\leq i\leq t-1$,  $\Psi(i) \geq 2$, and for $i\in\{t,\dotsc,n-1\}$, $\Psi(i) = 1$. Thus, plugging these values in the left hand side of the above inequality we see that the left hand side is strictly greater than $2^n$. 
Hence, we have 
\[ 2^n < 2^{d-n-1}.\]
Therefore, we get that $d> 2n+1$, when $n\geq 2$. 

We further observe that $t$ is at most $\lceil\log_2 n\rceil$, because for all $i \in \{\lceil\log_2 n\rceil, \dotsc, n-1\}$, we have $\Psi(i) = 1$.
We now claim that using $t\leq \lceil\log_2 n\rceil$ and $d=2n+2$, the inequality~\eqref{eq:kraft-bound} is satisfied. 
Plugging the value of $d$ and using the fact that $\Psi(i)-1=0$ for $i \in\{t,\dotsc ,n-1\}$ we can rewrite inequality~\eqref{eq:kraft-bound} as follows
\begin{equation}\label{eq:computer}
\sum_{i=0}^{t-1}\binom{n}{i} 2^{\Psi(i)-1} + \sum_{i=t}^{n-1}\binom{n}{i} + \binom{n}{n}\frac{1}{2} \leq 2^{n+1}.
\end{equation}
Now using $\sum_{i=0}^{n+1}\binom{n+1}{i} = 2^{n+1}$, we obtain 
\[\sum_{i=0}^{t-1}\binom{n}{i} 2^{\Psi(i)-1} + \sum_{i=t}^{n-1}\binom{n}{i} + \binom{n}{n}\frac{1}{2} \leq \sum_{i=0}^{n+1}\binom{n+1}{i} .\]
Moving the second summand to the right hand side we have
\[\sum_{i=0}^{t-1}\binom{n}{i} 2^{\Psi(i)-1} +\binom{n}{n}\frac{1}{2}  \leq \sum_{i=0}^{n+1}\binom{n+1}{i} - \sum_{i=t}^{n-1}\binom{n}{i}.\]
Rewriting again we get 
\[\sum_{i=0}^{t-1}\binom{n}{i} 2^{\Psi(i)-1} +\binom{n}{n}\frac{1}{2}  \leq \sum_{i=0}^{t-1}\binom{n+1}{i} + \sum_{i=t}^{n+1}\binom{n+1}{i} - \sum_{i=t}^{n-1}\binom{n}{i}.\]
Further rewriting leads to 
\[\sum_{i=0}^{t-1}\binom{n}{i} 2^{\Psi(i)-1} +\binom{n}{n}\frac{1}{2}  \leq \sum_{i=0}^{t-1}\binom{n+1}{i} + \sum_{i=t}^{n-1}\left[\binom{n+1}{i} -\binom{n}{i}\right] +\binom{n+1}{n} + \binom{n+1}{n+1}.\]
Now using the Pascal's rule, $\binom{n+1}{i}-\binom{n}{i} = \binom{n}{i-1}$, to simplify the second summand on the right hand side we have
\[\sum_{i=0}^{t-1}\binom{n}{i} 2^{\Psi(i)-1} +\binom{n}{n}\frac{1}{2}  \leq \sum_{i=0}^{t-1}\binom{n+1}{i} + \sum_{i=t}^{n-1}\binom{n}{i-1} +\binom{n+1}{n} + \binom{n+1}{n+1}.\]
Rewriting and simplifying we obtain,
\[\sum_{i=0}^{t-1}\binom{n}{i} 2^{\Psi(i)-1} +\binom{n}{n}\frac{1}{2}  \leq \sum_{i=0}^{t-1}\binom{n+1}{i} + \sum_{i=t-1}^{n-2}\binom{n}{i} +n+1 + 1.\]
Rewriting the right hand side again gives us
\[\sum_{i=0}^{t-1}\binom{n}{i} 2^{\Psi(i)-1} +\binom{n}{n}\frac{1}{2}  \leq \sum_{i=0}^{t-1}\binom{n+1}{i} + \sum_{i=t-1}^{n}\binom{n}{i} + 1.\]
Moving the second summand on the left hand side to the right and simplifying we have 
\[\sum_{i=0}^{t-1}\binom{n}{i} 2^{\Psi(i)-1}  \leq \sum_{i=0}^{t-1}\binom{n+1}{i} + \sum_{i=t-1}^{n}\binom{n}{i} + \frac{1}{2} .\]
Now we first observe that the right hand side is strictly greater than $2^n$. Using $t \leq \lceil\log_2 n\rceil$ we now show that the left hand side is at most $2^{\log_2^2n + 2\log_2 n}$. 
Thus showing that the inequality is satisfied. 
Using the fact that $\Psi(i)-1 \leq \Psi(0)-1\leq \lceil\log_2(n+1)\rceil -1 \leq \log_2 n$ and simplifying the left hand side, we obtain the following upper bound on it   
\[n\sum_{i=0}^{t-1}\binom{n}{i}.\]
Now using $\sum_{i=0}^k\binom{n}{i} \leq 2^{\mathbb{H}(k/n)\cdot n}$ for $k/n \leq 1/2$ and $\mathbb{H}$ being the Shannon entropy function, we can further bound it by 
\[n \cdot 2^{\mathbb{H}(\frac{t-1}{n})\cdot n}.\]
Now since $t-1 \leq \log_2 n$, we obtain the following upper bound on it
\[2^{\log_2 n + \log_2^2 n + \log_2 n},\]
where we used $\mathbb{H}(\log_2 n /n)\cdot n \leq \log_2^2 n+ \log_2 n$.
For $n \in \{1,\ldots,38\}$ we verified eq.~\eqref{eq:computer} with a computer calculation, see Section~\ref{sec:computer}.
The fact that $n \geq \log_2^2 n+ 2\log_2 n$ for all $n \geq 39$ completes the proof.
\end{proof}


\section{Proofs for the hazard-free Karchmer-Wigderson game}
\label{sec:hazfreeKWgameproofs}
In this section we prove Theorem~\ref{thm:hf-kwthm}.
We split the proof into two lemmas.

\begin{lemma}
\label{lem:formula-protocol}
Let $f\colon\Bool^n\to\Bool$ be a Boolean function and $F$ be 
a hazard-free De Morgan formula computing it. Then, 
\[\cc(\kw{\un}_f) \leq \depth{}(F),\quad \mbox{ and }\quad \monorect(\kw{\un}_f) \leq \size{}(F).\]
\end{lemma} 
\begin{lemma}
\label{lem:protocol-formula}
Let $f\colon\Bool^n\to\Bool$ be a Boolean function and $\Pi$ be a protocol 
for the hazard-free $\kw{}$-game $\kw{\un}_f$. Then, 
\[\depth{\un}(f) \leq \cc(\Pi),\quad \mbox{ and }\quad 
\size{\un}(f) \leq \monorect(\Pi).\]
\end{lemma}

The proofs of Lemmas~\ref{lem:formula-protocol} and \ref{lem:protocol-formula} 
are natural generalizations of their corresponding counterparts in the original 
setting of the $\kw{}$-game.  
However, for the sake of completeness and to highlight the   
differences, we present the proofs below. 
\begin{proof}[Proof of Lemma~\ref{lem:formula-protocol}] 
(Formula to Protocol.) Let $F$ be a hazard-free De Morgan formula computing 
the function $f$. It suffices to show a protocol $\Pi$ solving $\kw{\un}_f$ 
such that $\cc(\Pi) \leq \depth{}(F)$ and 
$\monorect(\Pi) \leq \size{}(F)$. 
Let $\alpha\in\hfe{f}^{-1}(1)$ be Alice's input and 
$\beta\in\hfe{f}^{-1}(0)$ be Bob's input. 

In the protocol $\Pi$ both players keep track of a subformula $G$ of $F$ 
such that $G(\alpha) = 1$ and $G(\beta)=0$. We being at the root of $F$, i.e., $G=F$. 
By the hazard-free property of $F$, we have 
$G(\alpha)=F(\alpha)=\hfe{f}(\alpha)=1$ and $G(\beta)=F(\beta)=\hfe{f}(\beta)=0$. 
Now depending on the type of gate at the root of $G$, 
we decide which player sends the message. 

If $G = G_0 \vee G_1$, then it is Alice's turn. Observe that, 
since $G(\alpha) = 1$, 
there exists $i\in\Bool$ such that $G_i(\alpha) = 1$. 
Again by a similar reasoning, we also have $G_0(\beta) = G_1(\beta) = 0$. 
Thus, Alice can send a single bit to Bob 
indicating a child which evaluates to $1$. 
Then they both move to the corresponding subformula. 

If $G = G_0 \wedge G_1$, then it is Bob's turn. 
Again by a similar reasoning, 
$G(\alpha) =1$ and $G(\beta)=0$, we deduce that $G_0(\alpha) = G_1(\alpha) = 1$ and 
there exists $i\in\Bool$ such that $G_i(\beta) = 0$. 
Thus, Bob can send a single bit to Alice 
indicating a child which evaluates to $0$ and 
they both move to the corresponding subformula. 

If $G$ is a leaf of $F$ and is labeled with literal $x_i$ or $\neg x_i$ 
for some $i\in[n]$, then the protocol returns the coordinate $i$ as answer. 

Clearly, the number of bits exchanged on any input equals 
the depth of the leaf reached and the number of leaves in the protocol $\Pi$ equals the number of leaves in $F$. That is, $\cc(\Pi) = \depth{}(F)$ and 
$\monorect(\Pi) = \size{}(F)$. 
The correctness of the protocol follows from observing that 
when $G$ is a literal and $G(\alpha)=1$ and $G(\beta)=0$ then the variable 
corresponding to the literal has different (Boolean) values 
in the certificates $\alpha$ and $\beta$. 
\end{proof}

\begin{remark}
We note that the hazard-freeness of $F$ is crucially used in the above proof. Indeed, if $F$ was a De Morgan formula computing $f$ but \emph{with hazards}, then the players would not be able to play the game on hazardous inputs using this formula. For example, consider the non-hazard-free formula for $\MUX_1$ in Figure~\ref{fig:mux1-eg}(a) and suppose Alice's input is $(\un,1,1)$. Since the root node is $\vee$, it is Alice's turn. However, since the formula has a hazard at $(\un,1,1)$, Alice doesn't know which subformula evaluates to $1$ and hence the game gets stuck. 
\end{remark}

We now give the translation from protocols to formulas (Lemma~\ref{lem:protocol-formula}). We note that there is no difference in the proof with respect to the original setting. The only observation is that if we start with a protocol that solves hazard-free $\kw{}$-game then we obtain a hazard-free formula. Nevertheless for the sake of completeness we present the proof below. 

\begin{proof}[Proof of Lemma~\ref{lem:protocol-formula}]
(Protocol to Formula.) Let $\Pi$ be a protocol solving $\kw{\un}_f$. 
It suffices to show a hazard-free formula $F$ computing $f$ such that 
$\depth{}(F) \leq \cc(\Pi)$ and $\size{}(F) \leq \monorect(\Pi)$. 

Let $T$ be the protocol tree of $\Pi$. We will convert the tree $T$ 
into a hazard-free formula $F$ for~$f$. Every internal node in the 
tree is associated with a player whose turn it is to send the message. 
To obtain $F$ we first replace every internal node in $T$ as follows:
\begin{compactitem}
    \item If it is associated with Alice, then replace it with $\vee$,
    \item If it is associated with Bob, then replace it with $\wedge$.
\end{compactitem}
Now consider a leaf $\ell \in T$ and suppose that the output 
at this leaf is some coordinate $i \in [n]$. Let $A_\ell \times B_\ell$ be 
the set of inputs that reaches this leaf $\ell$.
This is a monochromatic combinatorial rectangle.
By definition of $\kw{\un}_f$,
exactly one of 
the following cases holds:
\begin{compactenum}
    \item[(\emph{i})] for all $\alpha \in A_\ell$, $\alpha_i = 1$ and for all $\beta \in B_\ell$, $\beta_i =0$, or 
    \item[(\emph{ii})] for all $\alpha \in A_\ell$, $\alpha_i = 0$ and for all $\beta \in B_\ell$, $\beta_i =1$.  
\end{compactenum}
In the first case we label the leaf $\ell$ in $F$ with the literal $x_i$, 
while in the second we label it with $\neg x_i$. 
Clearly the constructed formula $F$ has depth and size equal to $\cc(\Pi)$ 
and $\monorect(\Pi)$, respectively.
So it remains to argue the correctness of the transformation. 
That is, we need to verify that $F$ is indeed a hazard-free formula for $f$.
It suffices to show the following: 
\begin{quote}
for every node $v\in F$, 
the subformula $G$ rooted at $v$ satisfies $G(\alpha) = 1$ 
for all $\alpha \in A$ and $G(\beta)=0$ for all $\beta\in B$, 
where $A\times B$ is the set of the inputs 
that reach the node corresponding to $v$ in the protocol tree $T$. 
\end{quote}
The correctness now follows by applying this claim to the root of $F$ (note that for the root we have $A=\hfe{f}^{-1}(1)$ and $B=\hfe{f}^{-1}(0)$).
We prove the claim by induction on the depth of nodes in the formula. 

\noindent\textbf{Base case}: $\depth{} = 0$. The claim holds for the leaf nodes 
by our construction (i.e., by our choices of their labels).

\noindent\textbf{Induction step}: Suppose the claim holds for the children 
$v_0$ and $v_1$ of a certain node $v \in F$. We will now show that 
it also holds for $v$. Let $G, G_0, G_1$ be the subformulas rooted at 
$v, v_0, v_1$, respectively. 
We assume, w.l.o.g., that $G = G_0 \vee G_1$. (The other case being symmetric.)
Let $v',v'_0,v'_1 \in T$ be the corresponding nodes to $v,v_0,v_1 \in F$. 
Let $A \times B$ be the inputs reaching $v'$ in $T$. 
Since $G= G_0 \vee G_1$, it is Alice's turn to send the message. 
Therefore, Alice's message partitions $A$ into $A_0$ and $A_1$ 
such that $A_0 \times B$ is the inputs reaching $v'_0$ and 
$A_1\times B$ is the inputs reaching $v'_1$. 
Thus, by induction hypothesis, we have 
\begin{compactitem}
\item for all $\alpha \in A_0$, $G_0(\alpha) = 1$ and for all $\beta\in B$, $G_0(\beta)=0$, 
\item for all $\alpha \in A_1$, $G_1(\alpha) = 1$ and for all $\beta\in B$, $G_1(\beta)=0$. 
\end{compactitem}
This in turn implies that for all $\alpha \in A$, 
$G(\alpha) = G_0(\alpha) \vee G_1(\alpha) =1$ and for all $\beta\in B$, 
$G(\beta) = G_0(\beta) \vee G_1(\beta) =0$. 
\end{proof}
\begin{remark}
We remark that since $\Pi$ is a protocol solving $\kw{\un}_f$, we obtain a hazard-free formula computing $f$. More generally, if $\Pi$ would be a protocol solving the hazard-free $\kw{}$-game on the input space $A\times B$, where $f^{-1}(1) \subseteq A \subseteq \hfe{f}^{-1}(1)$ and 
$f^{-1}(0) \subseteq B \subseteq\hfe{f}^{-1}(0)$, then we will obtain a formula computing $f$ that is hazard-free only on the inputs in $A\uplus B$. (See Proposition~\ref{pro:hf-limited}.)
\end{remark}

\section{Limited hazard-freeness}\label{app:limited-hazard-free}
Avoiding all hazards can be very expensive (see \cite{IKLLMS19,Jukna21}) and sometimes is not needed, because we might have additional information about the input, for example when composing circuits.
Under the (physically realistic in that setting) assumption that the input contains at most one $\un$, the counter in \cite{FKLP:17} outputs the number of 1s in the input, encoded in binary Gray code with a single $\un$, such that both resolutions of the output correspond to numbers whose difference is 1. Note that if Gray codes are not used, then in the usual binary bit representation
one has to flip 4 bits to go from $7=0111_2$ to $8=1000_2$,
hence any hazard-free counter would have to output $\un\un\un\un$ on input $1111111\un$. When composing circuits, this additional information about the position of the $\un$ can be useful, which is shown for the task of sorting Gray code numbers in \cite{LM:16, BLM:17, BLM:18, BLM:20}.
Another interesting class of hazards that should be avoided are the inputs where the number of $\un$'s is bounded from above. This is the setting of $k$-bit hazard-freeness from \cite{IKLLMS19}.

Theorem~\ref{thm:hf-kwthm} as stated is not directly applicable in these settings, as it 
only characterizes hazard-free formulas, i.e., formulas that are hazard-free with respect to 
\emph{all} inputs.
However, we can also treat formulas that avoid only certain hazards.
\begin{proposition}\label{pro:hf-limited}
Given sets $A$ and $B$ with 
$f^{-1}(1) \subseteq A \subseteq \hfe{f}^{-1}(1)$ and 
$f^{-1}(0) \subseteq B \subseteq\hfe{f}^{-1}(0)$.
The hazard-free $\kw{}$-game where Alice gets input from 
$A$ and Bob gets input from $B$ characterizes formulas that is hazard-free on inputs in $A\uplus B$.
\end{proposition}
The proof is a straightforward generalization of the proof of Theorem~\ref{thm:hf-kwthm} (and in particular of Lemmas~\ref{lem:formula-protocol} and \ref{lem:protocol-formula}).

\begin{example}\label{ex:limited-hf}
Consider the function $\MUX_{2}(s_1, s_2, x_{00}, x_{01}, x_{10}, x_{11})$. We proved in Theorem~\ref{thm:mux-lower} that any hazard-free formula for $\MUX_{2}$ requires $17$ leaves. Suppose we only want to be hazard-free on inputs where at least one selector bit is stable. Note that every other possible hazard is covered by the eight prime implicants and eight prime implicates of $\MUX_{2}$ labeling the rows and columns of the following matrix. Therefore, we can obtain an improved upper bound for this task by showing that the following game has a protocol of size smaller than $17$.
\[
\begin{pNiceMatrix}[margin,first-row,first-col]
\CodeBefore
\rectanglecolor{red!10}{1-6}{3-8}
\rectanglecolor{blue!10}{6-1}{8-3}
\rectanglecolor{green!10}{4-1}{4-2}
\rectanglecolor{red!20}{5-2}{5-3}
\cellcolor{blue!20}{4-3}
\cellcolor{yellow!20}{5-1}
\rectanglecolor{red!30}{1-1}{2-2}
\rectanglecolor{red!30}{1-4}{2-4}
\cellcolor{green!30}{1-3,1-5}
\cellcolor{blue!30}{3-1,3-4}
\cellcolor{orange}{2-3,2-5}
\cellcolor{purple}{3-2,3-3,3-5}
\cellcolor{green!60}{4-4,4-6,4-7,6-4,6-6,6-7,7-4,7-6,7-7}
\cellcolor{blue!40}{5-5,5-8,7-5,7-8,8-5,8-8}
\cellcolor{yellow!30}{4-5,4-8,6-5,6-8}
\cellcolor{red!40}{5-4,5-6,8-4,8-6}
\cellcolor{orange!50}{5-7,8-7}
\Body
{} & 000\un\un\un & 0\un00\un\un & 01{\un}0{\un}{\un} & \un 00{\un}0{\un} & \un 1{\un}0{\un}0 & 10{\un}{\un}0{\un} & 1\un{\un}{\un}00 & 11{\un}{\un}{\un}0 \\
001{\un}{\un}{\un} & x_{00} & x_{00} & s_2 & x_{00} & s_2 & s_1 & s_1 & s_1, s_2 \\
0{\un}11\un\un & x_{00} & x_{00}, x_{01} & x_{01} & x_{00} & x_{01} & s_1 & s_1 & s_1 \\
01\un1\un\un & s_2 & x_{01} & x_{01} & s_2 & x_{01} & s_1, s_2 & s_1 & s_1 \\
\un 01\un1\un& x_{00} & x_{00} & s_2 & x_{00}, x_{10} &  s_2 & x_{10} & x_{10} & s_2 \\
\un 1\un1\un1 & s_2 & x_{01} & x_{01} & s_2 & x_{01}, x_{11} & s_2 & x_{11} & x_{11} \\
10\un\un1\un & s_1 & s_1 & s_1, s_2 & x_{10} & s_2 & x_{10} & x_{10} & s_2 \\
1\un\un\un11 & s_1 & s_1 & s_1 & x_{10} & x_{11} & x_{10} & x_{10}, x_{11} & x_{11} \\
11\un\un\un1 & s_1, s_2 & s_1 & s_1 & s_2 & x_{11} & s_2 & x_{11} & x_{11} \\
\end{pNiceMatrix}
\]
Indeed, the above colouring yields a protocol of size 16. We can also show that we cannot do better. Substitute $s_1 = s_2 = \bigl(\begin{smallmatrix} 1&0 \\ 0&0 \end{smallmatrix}\bigr)$ and $x_{00} = x_{01} = x_{10} = x_{11} = \bigl(\begin{smallmatrix} 0&0 \\ 0&1 \end{smallmatrix}\bigr)$. This yields the following $16 \times 16$ (block) matrix. Note that if the original communication matrix had a protocol that yields fewer than $16$ combinatorial rectangles, then this block matrix must have rank less than $16$ since all substituted matrices are rank one. However, the block matrix is full-rank, showing that $16$ leaves are required.

\[
\begin{pmatrix}
0 & 0 & 0 & 0 & 1 & 0 & 0 & 0 & 1 & 0 & 1 & 0 & 1 & 0 & 1 & 0 \\
0 & 1 & 0 & 1 & 0 & 0 & 0 & 1 & 0 & 0 & 0 & 0 & 0 & 0 & 0 & 0 \\
0 & 0 & 0 & 0 & 0 & 0 & 0 & 0 & 0 & 0 & 1 & 0 & 1 & 0 & 1 & 0 \\
0 & 1 & 0 & 1 & 0 & 1 & 0 & 1 & 0 & 1 & 0 & 0 & 0 & 0 & 0 & 0 \\
1 & 0 & 0 & 0 & 0 & 0 & 1 & 0 & 0 & 0 & 1 & 0 & 1 & 0 & 1 & 0 \\
0 & 0 & 0 & 1 & 0 & 1 & 0 & 0 & 0 & 1 & 0 & 0 & 0 & 0 & 0 & 0 \\
0 & 0 & 0 & 0 & 1 & 0 & 0 & 0 & 1 & 0 & 0 & 0 & 0 & 0 & 1 & 0 \\
0 & 1 & 0 & 1 & 0 & 0 & 0 & 1 & 0 & 0 & 0 & 1 & 0 & 1 & 0 & 0 \\
1 & 0 & 0 & 0 & 0 & 0 & 1 & 0 & 0 & 0 & 1 & 0 & 0 & 0 & 0 & 0 \\
0 & 0 & 0 & 1 & 0 & 1 & 0 & 0 & 0 & 1 & 0 & 0 & 0 & 1 & 0 & 1 \\
1 & 0 & 1 & 0 & 1 & 0 & 0 & 0 & 1 & 0 & 0 & 0 & 0 & 0 & 1 & 0 \\
0 & 0 & 0 & 0 & 0 & 0 & 0 & 1 & 0 & 0 & 0 & 1 & 0 & 1 & 0 & 0 \\
1 & 0 & 1 & 0 & 1 & 0 & 0 & 0 & 0 & 0 & 0 & 0 & 0 & 0 & 0 & 0 \\
0 & 0 & 0 & 0 & 0 & 0 & 0 & 1 & 0 & 1 & 0 & 1 & 0 & 1 & 0 & 1 \\
1 & 0 & 1 & 0 & 1 & 0 & 1 & 0 & 0 & 0 & 1 & 0 & 0 & 0 & 0 & 0 \\
0 & 0 & 0 & 0 & 0 & 0 & 0 & 0 & 0 & 1 & 0 & 0 & 0 & 1 & 0 & 1
\end{pmatrix}
\]
\end{example}

For a natural number $k$, a Boolean circuit $C$ is said to have a $k$-bit hazard if there exists $\alpha\in\Tri^n$ such that $C$ has hazard at $\alpha$ and the number of $\un$'s in $\alpha$ is at most $k$. 
We can also obtain a $k$-bit hazard-free construction similar to  \cite[Theorem 5.3]{IKLLMS19} using Proposition~\ref{pro:hf-limited}.  
\begin{theorem}\label{thm:k-bit-hz-free}
For any Boolean function $f\colon\Bool^n\to\Bool$, there exists a $k$-bit hazard-free formula of 
depth at most \(2\log\left(\sum_{i=0}^k\binom{n}{i}\right) + 2k + \depth{}(f)\) and 
size at most \(\left(\sum_{i=0}^k\binom{n}{i}\right)^2\cdot 4^k\cdot \size{}(f)\).
\end{theorem}
\begin{proof}
Following Proposition~\ref{pro:hf-limited} it suffices to give a protocol for the hazard-free $\kw{}$-game where Alice gets an implicant $\alpha$ with at most $k$ $\un$'s and Bob gets an implicate $\beta$ with at most $k$ $\un$'s.  The protocol is as follows:
\begin{enumerate}
    \item Alice sends the set $A\subseteq [n]$ of positions of the $\un$'s in $\alpha$. 
    \item Bob then sends the set $B\subseteq [n]$ of positions of the $\un$'s in $\beta$ and the bits in $\beta$ at all positions from $A\setminus B$.  
    \item Alice then sends the bits in $\alpha$ for the positions in $B\setminus A$. 
    \item For the positions in $A\cap B$, they decide beforehand to set them to a fixed constant, say $0$. Thus no communication is required. They now have an instance of the classical game $\kw{}_f$, which they solve optimally. 
\end{enumerate}
The correctness of the protocol follows easily. The total number of bits exchanged to reduce to $\kw{}_f$ is at most $2\log\left(\sum_{i=0}^k\binom{n}{i}\right) + 2k$. Alice and Bob decide beforehand that the bits for positions in $A\setminus B$ and $B\setminus A$ are sent in sorted (say, increasing) order of positions. Thus they need to exchange at most $2k$ bits to send the bits across.  
\end{proof}

\section{Hazard-free formula depth reduction}
\label{sec:depthreduction}
In this section we show that the standard depth reduction process for Boolean circuits and formulas \cite{Brent:1974} works in a hazard-free way.
Let $F$ be a hazard-free formula.
In the Boolean depth reduction process we take a gate that has roughly the same distance to the root as it has to its deepest leaf and we write
\[
F(x) = F'(G(x),x)
\]
where $G$ is the subformula of that gate. Now we observe that
\begin{equation}\label{eq:depthreduction}
F(x) = F'(G(x),x) = \MUX_1(G(x),F'(0,x),F'(1,x)) =: \tilde F(x)
\end{equation}
which can be used iteratively to convert a Boolean formula to logarithmic (in the size of $F$) depth. If the implementation of $\MUX_1$ is hazard-free, then this depth-reduction process preserves the hazard-freeness of a formula, as the next claim shows.
\begin{claim}
Assume that a hazard-free implementation of $\MUX_1$ is used in \eqref{eq:depthreduction}.
If $F$ is hazard-free, then $\tilde F$ is also hazard-free.
\end{claim}
\begin{proof}
Consider the case when $\tilde F(\alpha)=u$. It remains to show that $F(\alpha)=\un$.
The hazard-freeness of the multiplexer implementation implies that there are the following three cases:
\begin{compactitem}
\item $G(\alpha)=0$ and $F'(0,\alpha)=\un$
\item $G(\alpha)=1$ and $F'(1,\alpha)=\un$
\item $G(\alpha)=\un$ and $F'(0,\alpha)$ and $F'(1,\alpha)$ are not both equal to the same Boolean value (in fact, they potentially have value $\un$).
\end{compactitem}
In the first two cases, by definition of $F'$ and $G$ we see that $F(\alpha)=F'(G(\alpha),\alpha)=\un$.
In the third case we see that $F'(\un,\alpha)=\un$. Using the fact that $G(\alpha)=\un$ implies that $F(\alpha)=F'(G(\alpha),\alpha)=\un$.
\end{proof}

By using standard techniques for finding the subformula $G$ (See Lemma~{1.3} in \cite{J:12}), we can show that the balanced formula has depth $3\log_{3/2}(m) + O(1)$ and size $O(m^{2.92})$ where $m$ is the size of the original formula.

\section{Sagemath Source Code for the Finite Cases}
\label{sec:computer}
The following sagemath (version 9) code is used to verify Equation~\eqref{eq:computer} up to $n=38$.

\begin{minipage}{\textwidth}
\small
\begin{verbatim}
def blog(k):
  return float(log(k,2))

def Psi(i,n):
  return int(ceil( blog(2^i+n-i) ))-i

def t(n):
  return int(ceil(blog(n)))

def LHS(n):
  return sum([binomial(n,i)*2^(Psi(i,n)-1)   for i in [0..t(n)-1]])    +  \
    sum([binomial(n,i)  for i in [t(n)..n-1] ])   +  0.5

def RHS(n):
  return 2^(n+1)

for n in [1..38]:
  if RHS(n)>=LHS(n):
    print("ok")
  else:
    print("problem")
\end{verbatim}
\end{minipage}

\bibliographystyle{alpha}
\bibliography{main}

\end{document}